\documentclass[11pt, letterpaper]{amsart}
\usepackage{amsfonts,amsmath,amssymb}
\usepackage[foot]{amsaddr}
\usepackage{color}
\usepackage{enumerate}
\usepackage{appendix}
\usepackage{graphicx}

\newtheorem{thm}{Theorem}[section]

\newtheorem{lem}[thm]{Lemma}
\newtheorem{prop}[thm]{Proposition}
\theoremstyle{definition}
\newtheorem{defn}[thm]{Definition}

\newtheorem{ass}[thm]{Assumption}

\theoremstyle{remark}
\newtheorem{rem}[thm]{Remark}

\numberwithin{equation}{section}

\newcommand{\filt}{\mathbb{F}}
\newcommand{\filtg}{\mathbb{G}}

\newcommand{\prob}{\mathbb{P}}
\newcommand{\tprob}{\widetilde{\mathbb{P}}}
\newcommand{\qprob}{\mathbb{Q}}
\newcommand{\reals}{\mathbb R}
\newcommand{\symmat}{\mathbb S}

\newcommand{\sdpos}{\mathbb{S}^d_{++}}
\newcommand{\A}{\mathcal{A}}
\newcommand{\B}{\mathcal{B}}

\newcommand{\E}{\mathcal{E}}
\newcommand{\F}{\mathcal{F}}
\newcommand{\G}{\mathcal{G}}

\newcommand{\M}{\mathcal{M}}

\newcommand{\mcp}{\mathcal{P}}

\newcommand{\V}{\mathcal{V}}
\newcommand{\We}{\mathcal{W}}

\newcommand{\eps}{\varepsilon}
\newcommand{\such}{\ | \ }

\newcommand{\probtriple}{(\Omega, \mathcal{F}, \mathbb{P})}

\newcommand{\Leb}{\textrm{Leb}}
\newcommand{\dfn}{\, := \,}
\newcommand{\rdfn}{\, =: \,}

\newcommand{\indep}{\perp\!\!\!\perp}
\newcommand{\pvct}{\varphi}
\newcommand{\pmtx}{\chi}
\newcommand{\expv}[3]{\mathbb{E}^{#1}_{#2}\left[#3\right]}

\newcommand{\condexpv}[4]{\mathbb{E}^{#1}_{#2}\left[#3\big| #4\right]}

\newcommand{\expvs}[1]{\mathbb{E}\left[#1\right]}
\newcommand{\condexpvs}[2]{\mathbb{E}\left[#1\big| #2\right]}

\newcommand{\condprobs}[2]{\mathbb{P}\left[#1\big| #2\right]}

\newcommand{\nada}[1]{}

\newcommand{\wt}[1]{\widetilde{#1}}
\newcommand{\wh}[1]{\widehat{#1}}

\newcommand{\bra}[1]{\left[#1\right]}
\newcommand{\cbra}[1]{\left\{#1\right\}}
\newcommand{\ol}[1]{\overline{#1}}

\begin{document}

\title[Equilibrium with Flows]{Equilibrium with Heterogeneous Information Flows}

\author{Scott Robertson}
\address{Questrom School of Business, Boston, University, Boston MA, 02215, USA}
\email{scottrob@bu.edu}

\date{\today }

\begin{abstract}
   We study a continuous time economy where throughout time, insiders receive private signals regarding the risky assets' terminal payoff.  We prove existence of a partial communication equilibrium  where, at each private signal time, the public receives a signal of the same form as the associated insider, but of lower quality.  This causes a jump in both the public information flow and equilibrium asset price.  The resultant markets, while complete between each jump time, are incomplete over each jump.  After establishing equilibrium for a finite number of private signal times, we consider the limit as the private signals become more and more frequent.  Under appropriate scaling we prove convergence of the public filtration to the natural filtration generated by both the fundamental factor process $X$ and a continuous time process $J$ taking the form $J_t = X_1 + Y_t$ where $X_1$ is the terminal payoff and $Y$ an independent Gaussian process.  This coincides with the filtration considered in \cite{MR2213260}. However, while therein the filtration was exogenously assumed to be that of an insider who observes a private signal flow, here it arises endogenously as the public filtration when there are a large number of insiders receiving signals throughout time.

\end{abstract}

\maketitle

\section{Introduction}  In this paper, we study equilibria in a multi-asset, multi-agent continuous time economy where there is asymmetric information.  The novelty is that we allow for agents with private information to dynamically enter the market.  This enables us to identify how the market reacts and equilibriates upon the arrival of new information.

The notion that prices convey private information dates at least back to \cite{hayek1945use}, and formalizing this notion is a classic problem in the asset pricing literature.  The earliest works are \cite{grossman1976efficiency, grossman1980impossibility}, which consider a one period model where the asset payoff and private signals are jointly normal, all agents have preferences described by an exponential (CARA) utility function, and the economy is competitive in that agents take prices as given.

\cite{grossman1976efficiency} showed that prices are fully revealing, in that they provide as much information regarding the risky asset terminal payoff as the private signals. This created a paradox, as fully revealing prices remove the incentive to be an insider, but the price itself depends on the presence of insiders. \cite{grossman1980impossibility, hellwig1980aggregation, diamond1981information} avoid the paradox by introducing ``noise'' traders (agents with inelastic exogenous demand). Additionally, \cite{grossman1980impossibility} introduces an uninformed agent who is rational, but does not have a private signal.

After these early works, many papers extended the basic model.  While a complete literature review is beyond the scope of this paper (see \cite{garcia2015asymmetric} for a partial overview), in the case of a competitive economy we highlight  \cite{brennan1996information} which allowed for agents to receive flows of private signals in a discrete time\footnote{\cite{brennan1996information} also allowed the trading periods to increase, informally obtaining a limiting continuous time economy.}; \cite{han2018dynamic} which considered when private information may be dynamically obtained at a cost; and  \cite{MR4192554} which allowed for continuous trading in a diffusive economy where a single signal is received at time zero.

Alternatively, a strand of literature arose which allowed the insider to perceive her impact on prices. The idea is that an agent who is influential enough to obtain a private signal might also be large enough to move prices, and should account for price impact.  In the seminal \cite{kyle1985continuous}, the insider is risk neutral and receives a perfect signal (i.e. with no noise) about the asset terminal payoff. She accounts for price impact by recognizing her demand will be combined with the noise traders' demand when submitted to the market maker (also risk neutral) who then quotes prices.   As in the price-taking literature, all random variables are jointly normal.

\cite{kyle1985continuous} has been extended in numerous directions, and as with the price taking case, a full literature review is beyond our scope. Rather, we selectively highlight extensions to when the insider sees the noise trader's demand prior to submitting her order (\cite{rochet1994insider}); to continuous time (\cite{back1992insider}); to when multiple insider interact (for noisy signals see \cite{foster1996strategic}, and for perfect signals see \cite{holden1992long} in discrete time and \cite{back2000imperfect} in continuous time); and to when the insiders and/or market makers have CARA preferences (see \cite{subrahmanyam1991risk} in one period, and the recent \cite{bose2023kyle, ccetin2022order} in continuous time).

Common to all the above price-impact papers is that private signals occur only at time $0$.  As such, the question of how an economy reacts to the onset of new private information is left unanswered.  Our paper allows one to see, in a competitive economy, how markets change with the arrival of new information, and provides a benchmark to which price impact analysis could compare, should equilibrium models accounting for price impact and dynamic private signal arrival be developed in the future.

In our model, there are $d$ risky assets $S$ with terminal (time $1$) payoff $S_1 = X_1$, where $X$ is a $d$-dimensional Ornstein-Uhlenbeck (OU) process driven by a Brownian motion $B$.\footnote{We discuss extending to general diffusions in Remark \ref{R:why_ou}.}  $X$ is publicly observable so all agents have access to the information flow $\filt^B$. There are $N$ insiders, and the $n^{th}$ insider receives a private signal $G_n$ at time $t_n$ where $0 = t_1 < t_2< ... < t_N < 1$.  The signals take the form $G_n = X_1 + Y_n$ where  $\cbra{Y_n}$ are independent Gaussian random vectors.   There is also an uninformed agent with no private information, and noise traders with exogenously specified Gaussian demands. The uniformed agent trades continuously over $[0,1]$, while the $n^{th}$ insider trades over $[t_n,1]$\footnote{This conforms the insider waiting to ``clean'' or refine her signal before trading. In Section \ref{SS:two} we briefly discuss what happens if one allows insiders to trade prior to receiving their signal.}. Agents have CARA preferences and derive utility over terminal wealth.

In this environment, we expect a ``shock'' to the market at each signal time $t_n$ as new private information enters which, as we show, causes a jump in both the equilibrium price $S$ and public information flow $\filt^m$.  This filtration enlarges $\filt^B$, and the determination of $\filt^m$, along with $S$, is part of the equilibrium problem.   Due to the jumps, we also expect the resultant market is incomplete, as new information cannot be perfectly hedged prior to its arrival.  

The equilibrium is constructed in  Theorem \ref{T:main_result}. In it, at each signal time $t_n$ the public, comprising of the uninformed agent as well as insiders $k=1,...n-1$, receives the signal $H_n = G_n + $``noise'' (see \eqref{E:H_def}).  Therefore, there is a public flow of signals, with each signal being of lesser quality than the associated insider signal.  This implies (see Section \ref{SS:two} for a qualitative explanation when there are $N=2$ insiders)  the public filtration is $\filt^m = \filt^B \vee \sigma(H_1,H_2,...,H_n)$ over the period $[t_n,t_{n+1})$, and we show the $\cbra{H_n}$ are such that one may use standard results on initial enlargements (see \cite{Jacod1985}, \cite{MR884713}).  The $k^{th}$ insider's filtration is $\filt^B\vee \sigma(H_1,...,H_n, G_k)$, and the resultant filtrations satisfy the usual conditions.

The equilibrium price process $S$ jumps at each signal time, and is diffusive in between signal times.  Due to the exponential-Gaussian nature of the problem, the price process is affine in the both the factor process $X$ and relevant signals (e.g. $(H_1,...,H_n)$ over $[t_n,t_{n+1})$) with time varying coefficients.  Additionally, $S$ has non-degenerate volatility and hence the market is endogenously complete over the intervals of continuous trading.

Equilibrium trading strategies are affine in the price process and relevant public signals, and take a very simple form: see \eqref{E:strats}.  In particular, each insider takes a static position in her private signal and it is this static position, when combined with the noise traders' demand, that both creates the public signal and communicates the public signal to the market. This in turn shows the equilibrium satisfies the rational expectations property introduced in \cite{radner1979rational} and extended in \cite{MR4192554}: see Definition \ref{D:DNREE}.

To establish equilibrium, we work backwards.  This leads to $2N-1$ ``periods'' to the model: $N$ intervals of continuous trading, and $N-1$ jumps when the signals arrive.  Over the last period $[t_N,1]$ of continuous trading, we establish equilibria using arguments very similar to those in \cite{MR4192554}. To establish equilibrium over the remaining periods we exploit the OU structure which yields effective endowments implied by future trading that are linear-quadratic in the price process, factor process and signals. However, even with the OU process, establishing equilibrium is delicate as agents have differing information sets.  To put things on common ground we first consider a ``signal-realization'' equilibrium where agents have common information but random endowments parameterized by signal realizations.  Here well known results, summarized in Propositions \ref{P:abstract_cont_time}, \ref{P:abstract_sp},  yield an equilibrium.

Having established equilibrium at the signal realization level, we then pass to the signal level using the martingale preserving measure as introduced in \cite{MR1418248, MR1632213}.  Here, a delicate issue arises when constructing the acceptable trading strategy sets because the signals, even though they are random variables, are seen by the agents and should be thought of as ``constants''.  To account for this, we follow \cite{detemple2022dynamic} and define the acceptable strategies at the signal level, as those which conform to the standard definition of acceptability for exponential utility (c.f. \cite{MR1891730, MR2489605}) at the signal realization level, for almost all realizations. With everything properly defined, we follow an inductive argument in Section \ref{S:pce_construct} to establish the equilibrium of Theorem \ref{T:main_result}.

In Section \ref{S:One_jump} we provide a numerical example in the two insider case. First, in Figure \ref{F:Trading_SP} we demonstrate the jump in the asset price, trading strategies, and market price of risk when the second signal enters. Next, we consider agent welfare, which following \cite{L-1985} is defined in terms of ex-ante certainty equivalents.  Here, we stress that a proper welfare analysis deserves its own treatment. This is because the key question ``Is is better to receive a worse signal now, or a better signal later?'' requires careful modeling of the relationship between signal time and signal precision.  Indeed, as shown in Figure \ref{F:Welfare}, based upon one's modeling choice, one may come to a variety of different, contradicting, conclusions. To overcome this, we believe additional criteria must be layered onto the welfare analysis to endogenize the signal time to signal precision relationship. However, this is topic for future research.

Having considered the case of finitely many signals, in Section \ref{S:CJ} we investigate what happens as the signals become more and more frequent. Here, we are motivated by both computational considerations and \cite{MR2213260} which considered when the insider has access to a flow $G=\cbra{G_t}$ of private signals taking the form (using our notation) $G_t = G(X_1, Y_t)$ where $Y$ is an independent  process. \cite{MR2213260} took this information flow, as well as the risky asset price process, as given and then considered the optimal investment problem for an agent with log-utility. It is natural to wonder if such an information structure can arise endogenously in equilibrium, and in Proposition \ref{P:filt_converge} and Theorem \ref{T:limit_bm} we show that it is indeed possible under the scaling assumptions in Assumption \ref{A:converge}.

To prove this result, we embed the finite signal analysis into a sequence of markets where signals arrive at a function $\tau$ of the dyadic rationals.  Under appropriate scaling, the market filtrations $\cbra{\filt^{N,m}}$ established at the discrete level, converge, in the sense of \cite{MR1113218, MR1739298, MR1837294}, to the filtration $\filt^{B,J}$, where $J_t = X_1 + Y_t$ where $Y$ is Gaussian process, whose mean and covariance function is determined entirely by the time function $\tau$ as well as the precision weighted risk tolerance function $\lambda$ of \eqref{E:ll_def}. To prove this result, we adapt results found in \cite{neufcourtthesis}.  Lastly, we identify the ``information drift'' (c.f. \cite{MR1940496}) process $\theta$ so that $B_\cdot - \int_0^\cdot \theta_u$ is a Brownian motion in the limiting filtration.

This paper is organized as follows: the probabilistic and financial setup, along with the main result, is given in Section \ref{S:model}. Section \ref{S:ExpCaraEq} provides the core exponential-Gaussian equilibrium results, both over a single period and in continuous time, when agents have common information.  Section \ref{S:EqEst} identifies the signal density processes which verify the relevant Jacod (c.f. \cite{Jacod1985}) conditions, and allow us to use standard results on filtration enlargements.  Section \ref{S:pce_construct} establishes equilibria. Section \ref{S:One_jump} provides a numerical example when there are two signals.  Section \ref{S:CJ} gives the convergence argument as the signals become more frequent.  Appendix \ref{AS:strat} contains more information on the acceptable trading strategies, and proofs are in Appendices \ref{AS:ExpCaraEq} -- \ref{AS:CJ}.

\section{The Model}\label{S:model}

\subsection*{The public factor process and traded assets}

There is a probability space $\probtriple$ supporting a $d$-dimensional Brownian motion $B$, and $\filt^B$ is the $\prob$-augmented natural filtration for $B$. The fundamental process $X$ a $d$-dimensional OU process with dynamics
\begin{equation*}
dX_t = \kappa(\theta-X_t)dt + \sigma dB_t;\qquad \Sigma \dfn \sigma\sigma',
\end{equation*}
where $\kappa \in\reals^{d\times d}$, $\theta\in \reals^d$, and $\Sigma \in \mathbb{S}^d_{++}$ the space of strictly positive definite $d\times d$ symmetric matrices. This implies (\cite[Chapter 4]{MR2053476}) for $s < t$ and given $F^B_s$, $X_t$ is normally distributed with mean vector $\mu(X_s,t-s)$, and covariance/precision matrices $\Sigma(t-s), P(t-s)$ where
\begin{equation}\label{E:X_trans_density}
\mu(x,\tau) \dfn \theta +  e^{-\tau\kappa}(x-\theta);\quad \Sigma(\tau) \dfn \int_0^\tau e^{-(\tau-u)\kappa}\Sigma e^{-(\tau-u)\kappa'}du;\quad P(\tau)\dfn \Sigma(\tau)^{-1}.
\end{equation}
The investment time horizon is $T=1$. There are $d$ risky assets $S$ with terminal payoff $S_1 = X_1$ and outstanding supply $\Pi\in\reals^d$, as well as a riskless asset in $0$ net supply, with price process exogenously fixed at $1$. The economy is competitive, so all participating agents take the price process $S$ as given.

\subsection*{Agents}

There are two agent types: uninformed and insider.   The uninformed agent, labeled $0$, receives no private signal.  He trades over the entire horizon $[0,1]$, and his information is obtained from publicly available quantities such as the factor process $X$, the price process $S$, and combined trading polices of all other participants.

Conversely, there are $N$ insiders. For $n=1,...,N$, the $n^{th}$ insider receives a private signal $G_n$ at a time $t_n \in [0,1)$, and does not participate in the market until she receives her signal.  The signal times are $0 = t_0 < t_1 < ... < t_N < 1$, so that the first private signal is received at time $0$.  The signals $\cbra{G_n}$ are noisy versions of the asset terminal payoff $S_1 = X_1$,
\begin{equation}\label{E:Gn_def}
G_n = X_1 + Y_n,\quad Y_n \sim N(0,C_n^{-1})\qquad n=1,...,N,
\end{equation}
so that the precision matrix $C_n$ measures the quality of the signal.   We assume the insiders use separate channels for obtaining their private information, so the $\cbra{Y_n}$ are independent of each other, as well as the terminal public information $\F^B_1$.

\begin{rem}
One may allow the $n^{th}$ insider to trade as an uninformed agent prior to $t_n$, and the equilibrium information structure does not change.  However, other equilibrium quantities such as prices and strategies will change, and the notation, already quite cumbersome, becomes even more complicated. Additionally, not trading prior to $t_n$ reflects that insiders take time to refine their signal before making investment decisions.  This being said, a discussion on what happens with $2$ insiders if the second is allowed to trade over $[0,t_2]$ is presented in Section \ref{SS:two}.
\end{rem}

Agents have  CARA (exponential) preferences, derive utility from terminal wealth, and do not have random endowments.  For $k=0,...,N$, agent $k$ has absolute risk aversion $\gamma_k$ and, to justify the price taking assumption, weight in the economy of $\omega_k$. Because it plays a key role in the analysis below, define the weighted risk tolerance
\begin{equation}\label{E:alpha_def}
\begin{split}
\alpha_k \dfn \frac{\omega_k}{\gamma_k},
\end{split}
\end{equation}
and, for $n=1,...,N$, the representative risk aversion among agents $k=0,...,n$
\begin{equation}\label{E:olgamma_def}
\ol{\gamma}_n \dfn \left(\sum_{k=0}^n \alpha_k\right)^{-1}.
\end{equation}

\subsection*{Noise Trading} Following the standard convention (see \cite{hellwig1980aggregation, kyle1985continuous, back1992insider}), to prevent full signal revelation we assume over $(t_n,t_{n+1}]$ ($[t_1,t_2]$ if $n=1$) there are noise traders with price inelastic demand. This demand is held constant over the period, and takes the form $\sum_{k=1}^n Z_k$.  Here, $Z_n \sim N(0,D_n^{-1})$, so that $D_n$ measures the precision, or volume (see \cite{KOVALENKOV2014211,NEZAFAT2023105664}), of noise trading. We assume the $\cbra{Z_n}$ are independent of each other, the insider dispersions $\cbra{Y_n}$, as well as the terminal public information $\F^B_1$.  Effectively, the noise traders enforce a random supply change of $Z_n$ at $t_n$.

\section{Equilibrium Definition and the Main Result}\label{S:main_result}

Below, we formally define and establish equilibrium for the case of $N$ insiders.  However, to provide intuition, we first analyze when there are two insiders, as this is the main departure from existing results. A numerical example highlighting the two insider case is given in Section \ref{S:One_jump}.

\subsection{Two insiders}\label{SS:two} At time $t_1 = 0$ the first insider receives the signal $G_1$ and the noise trader demand is $Z_1$.  Thus, if any information is to be communicated to the market at this time, it must be of the form $H_1 = H_1(G_1,Z_1)$ because there is no other uncertainty. Next, due to our assumption of exponential preferences and Gaussian dispersions\footnote{Though not necessarily a Gaussian factor process - see \cite{MR4192554}.}, it turns out the insider's optimal trading policy is \emph{static} in the private signal (there is also a dynamic component to the strategy, but this only depends on public information) and hence, as the noise trader demand $Z_n$ is also static, no additional information is released to the market. This implies that over $[0,t_2)$ one expects the public, or ``market'', information to be
\begin{equation*}
\F^m_t = \F^B_t \vee \sigma(H_1);\quad 0\leq t < t_2.
\end{equation*}
Additionally, by analyzing the clearing condition $\Pi = \omega_0 \wh{\pi}^0_t + \omega_1 \wh{\pi}^1_t + Z_1$, we see that if $\wh{\pi}^1_t$ is the sum of a static position in $G_1$ and a dynamic position which only uses only public information, the public signal $H_1$ must be a linear combination of $G_1$ and $Z_1$ (see \eqref{E:H_def} below) and hence by scaling we can view $H_1$ as a lower quality version of the insider's signal $G_1$.  Also, the publicly observable
\begin{equation*}
\Pi - \omega_0 \wh{\pi}^0_0 = \omega_1 \wh{\pi}^1_0 + Z_1
\end{equation*}
reveals $H_1$, and this provides a mechanism for how $H_1$ is passed to the market. The equilibrium asset price $S$ is adapted to $\filt^m$, hence continuous as only $\filt^B$ is changing with time, with value $S_{(t_2)-}$ right before the second signal time $t_2$.

Now, consider what happens at $t_2$. The second insider arrives with private signal $G_2$ and the noise trader demand changes by $Z_2$.  Thus, it is natural to assume a signal of the form $H_2 = H_2(G_2,Z_2)$ is passed to the market.  However, as $G_2 = X_1 + Y_2$ and $(Z_2,Y_2)$ are independent of $\F^B_{t_2}\vee \sigma(H_1)$, it must be that the public filtration experiences a jump at $t_2$.  This will in turn cause a jump in equilibrium quantities such as the price process.  Furthermore, as the CARA/Gaussian setting implies the second insider's position is also static in $G_2$, no additional information is communicated to the market after $t_2$, and hence over $[t_2,1]$ one expects the public information to be
\begin{equation*}
\F^m_t = \F^B_t \vee \sigma(H_1,H_2);\quad t_2\leq t \leq 1.
\end{equation*}
The clearing condition ensures $H_2$ is linear in $(G_2,Z_2)$ so that $H_2$ is a worse version of $G_2$, and $H_2$ is revealed to the market through the combined insider and noise trader demand $\Pi - \omega_0\wh{\pi}_{(t_2)+}$ immediately after $t_2$. Over $[t_2,1]$ the asset will be diffusive, taking the value $S_1 = X_1$ at the end.

The uninformed agent and first insider are aware the second insider will arrive at $t_2$, and that the noise trader demand will change (though not aware of the exact change amounts). Thus, they anticipate a jump at $t_2$, and optimally invest over the jump.  This yields essentially three equilibrium problems: one where agents $0,1$ trade continuously in $S$ over $[0,t_2)$; one where agents $0,1$ trade over the jump at $t_2$; and one where agents $0,1,2$ trade continuously in $S$ over $[t_2,1]$.  Market clearing must hold in all three problems; in particular at the jump $t_2$, so the clearing condition must be met even at a set of times with null Lebesgue measure.

Lastly, as alluded to above, consider when agent $2$ is allowed to trade prior to $t_2$, and write her effective endowment at $t_2$ from optimally trading over $[t_2,1]$ as (see \eqref{E:checkE_tN_def}) $\wt{\E}^2_{t_2}$. Note that in addition to $\F^B_{t_2}$,  after the jump at $t_2$ agent $2$ sees $(H_1,H_2,G_2)$ and agent $0$ sees $(H_1,H_2)$, while immediately before the jump agents $0,2$ both see $H_1$.  Remarkably, if agent $2$ first ``conditions away'' $G_2$ in her endowment $\wt{\E}^2_{t_2}$ by computing
\begin{equation*}
\gamma_2\ol{\E}^2_{t_2} \dfn - \log \condexpvs{e^{-\gamma_2\wt{\E}^2_{t_2}}}{\F^B_{t_2}\vee\sigma(H_1,H_2)},
\end{equation*}
then up to a constant, $\ol{\E}^2_{t_2}$ coincides with $\wt{\E}^0_{t_2}$, agent $0$'s effective endowment.  Therefore, over the jump at $t_2$ and over $[0,t_2)$, for purposes of establishing equilibrium quantities, agents $0$ and $2$ can be combined into a single agent with weight $\omega_0+\omega_1$ and risk aversion $\gamma_{02}$ where $1/\gamma_{02} = \omega_0/(\omega_0+\omega_2)\times 1/\gamma_0 + \omega_2/(\omega_0+\omega_2) \times 1/\gamma_2$. However, this change still affects the time $(t_2)-$ price, and hence all quantities over $[0,t_2)$. In the model, we do not allow agent $2$ to trade, preferring to think she waits while cleaning her signal.

\subsection*{Equilibrium Definition} With the above as intuition, we now consider $N$ insiders, using the notion of a ``partial communication equilibrium'' (PCE) as introduced in \cite{MR4192554}. Informally, a PCE is a pair $(\filt^m,S)$, where the market filtration $\filt^m$ used by the uninformed agent satisfies the usual conditions, and strictly enlarges $\filt^B$; $S$ is a RCLL $\filt^m$ semimartingale with $S_1 = X_1$ almost surely; and $(\filt^m,S)$ are such that the market clears with everyone optimally trading in $S$ trading over their respective periods, and using their respective information sets.  Additionally, to rule full revelation of the insider's signal, we  want $\filt^m \subset \filt^m\vee \sigma(G_k)$ over $[t_k,1]$ for $k=1,...,N$.

To identify $\filt^m$, we are motivated by the two-insider case, as well as previous results in the literature (see for example \cite{grossman1976efficiency, grossman1980impossibility, hellwig1980aggregation, admati1985noisy, breon2015existence, MR4192554}), and assume there are market signals $H_1,H_2,...,H_N$ transferred to the public at the respective signal times $t_1,...t_N$.   The $\cbra{H_n}$ are precisely defined in \eqref{E:H_def}. Then, for fixed $n$ and $1\leq k\leq n$, we set
\begin{equation}\label{E:filtnk_def}
    \filt^{n,0} \dfn \filt^B \vee \sigma(\ol{H}_n);\quad \filt^{n,k}\dfn \filt^B \vee \sigma(\ol{H}_n,G_k);\quad  \ol{H}_n \dfn (H_1,....,H_n).
\end{equation}
We  conjecture for $t\in [t_n,t_{n+1})$ the market and $k^{th}$ insider information sets are respectively
\begin{equation}\label{E:filtm_ik_def}
    \F^m_t = \F^{n,0}_t, \qquad \F^{I_k}_t = \F^{n,k}_t.
\end{equation}
To ensure the $\cbra{\filt^{n,k}}$ satisfy the usual conditions (see \cite[Proposition 3.3]{MR1775229} and \cite[Lemma 4.2]{MR3758346}), and to preserve the semi-martingale property going from both $\filt^B$ to $\filt^{m}$ and from $\filt^{m}$ to $\filt^{I_k}$, we assume the respective Jacod (c.f. \cite{Jacod1985}) conditions
\begin{equation*}
    \begin{split}
        \condprobs{\ol{H}_n \in d\ol{h}_n}{\F^B_t} &\sim \prob\bra{\ol{H}_n\in d\ol{h}_n} \textrm{ and }\condprobs{G_k \in dg_k}{\F^{n,0}_t} \sim \prob\bra{G_k\in dg_k} \quad \textrm{a.s., } t\in [0,1].
    \end{split}
\end{equation*}
Because $\filt^m$ ``jumps''  at each time $t_2,...,t_N$ from $\filt^B\vee\sigma(\ol{H}_n)$ to $\filt^B \vee \sigma(\ol{H}_{n+1})$, we expect the price process $S$ to be  diffusive over $[t_n,t_{n+1})$, with jumps at $t_2,...,t_N$. Accordingly we define the measure $\nu \dfn \Leb_{[0,1]} + \sum_{n=2}^N \delta_{t_n}$, and expect $d[S,S] << \nu$ almost surely (we only use $\nu$ to state the clearing condition).

To define the optimal investment problems, write $\A^0$ as the uninformed agent's set of acceptable trading strategies.   While a precise definition is given in Section \ref{S:pce_construct} and Appendix \ref{AS:strat} below, $\A^0$ contains those $\filt^m$ predictable, $S$ integrable policies such that the gains process
\begin{equation}\label{E:wealth_process_def}
    \We^{\pi}_{0,\cdot} \dfn \int_0^\cdot \pi_u'dS_u,
\end{equation}
satisfies an appropriate budget constraint.  Next, for $k=1,...,N$ write $\A^{k}$ as the $k^{th}$ insider's set of acceptable trading strategies.  These are $\filt^{I_k}$ predictable, $S$ integrable polices also such that $\We^{\pi}$ satisfies an appropriate budget constraint.  As the initial wealth factors out, the optimal investment problems are
\begin{equation}\label{E:vf_def}
    \begin{split}
        &\inf_{\pi\in \A^0} \condexpvs{e^{-\gamma_0 \We^{\pi}_{0,1}}}{\F^m_0} \textrm{ (Uninformed) };\qquad \inf_{\pi\in \A^k} \condexpvs{e^{-\gamma_k \We^{\pi}_{t_k,1}}}{\F^{I_k}_{t_k}} \ (k^{th}\textrm{ Insider) }.
    \end{split}
\end{equation}
We can now define the PCE.
\begin{defn}\label{D:pce}
Let $\filt^m$ be defined as in \eqref{E:filtnk_def}, \eqref{E:filtm_ik_def}, and let $S$ be a RCLL $\filt^m$ semi-martingale with $S_1 = X_1$ almost surely. Then the pair $(\filt^m,S)$ is a PCE if $\sigma(\ol{H}_n)\subset \sigma(\ol{H}_n,G_k)$ for $k=1,...,N$, and if there are optimizers $\cbra{\wh{\pi}^k}$ to the optimal investment problems in \eqref{E:vf_def} such that for $n=1,...,N$,
\begin{equation*}
    \Pi = \omega_0 \wh{\pi}^0 + \sum_{k=1}^n \left(\omega_k\wh{\pi}^k + Z_k\right),
\end{equation*}
$\prob\times \nu$ almost surely over $(t_n,t_{n+1}]$.
\end{defn}

\begin{rem} The $\prob\times\nu$ almost sure statement simply means the clearing condition must hold, for the participating agents, Lebesgue almost surely in between jump times, and at each jump time.
\end{rem}

The PCE definition does not explicitly state \emph{how} the signals $\cbra{H_n}$ are transferred to market. Guided by the two-signal discussion, as in \cite{MR4192554,detemple2022dynamic} we now define a dynamic noisy rational expectations equilibrium (DNREE) (see also \cite{radner1979rational})\footnote{Comparing to the REE of \cite{radner1979rational}, ``dynamic'' is due to continuous trading, and ``noisy'' is due to noise traders.}. To state it, first define
\begin{equation}\label{E:resid_demand}
\wh{\pi}^{R}_t \dfn \sum_{k=1}^{n-1} \left(\omega_k \wh{\pi}_{t}^k + Z_k\right) \textrm{  if  } t\in (t_{n-1},t_n];\qquad \wh{\pi}^R_0 = \omega_1 \wh{\pi}^1_0 + Z_1,
\end{equation}
so that $\wh{\pi}^R$ is the combined demand of everyone but the uninformed agent.

\begin{defn}\label{D:DNREE}
The pair $(S,\filt^m)$ is a DNREE if $(S,\filt^m)$ is a PCE, and if $\filt^m$ coincides with $\filt^{Obs}$, the $\prob$-augmentation of the right-continuous enlargement of the filtration generated by $B,S, \wh{\pi}^R$.
\end{defn}

A DNREE thus prevails if market observable quantities transmit the signals. Note also that by right-continuity of $\filt^m$  we already know $\filt^{Obs} \subseteq \filt^m$. To construct the PCE, recall \eqref{E:Gn_def}, \eqref{E:alpha_def}, and define the candidate market signals
\begin{equation}\label{E:H_def}
    H_n \dfn G_n + \frac{1}{\alpha_{n}}C^{-1}_n Z_n.
    \end{equation}
$H_n$ takes the same form as $G_n$ but is of lower quality as
\begin{equation}\label{E:H_def2}
H_n = X_1  + N(0,E^{-1}_n);\qquad E^{-1}_n \dfn C^{-1}_n + \frac{1}{\alpha_{n}^2}C_n^{-1}D_n ^{-1}C_n^{-1}.
\end{equation}
With the candidate signals defined (recall also \eqref{E:olgamma_def}), we now state
\begin{thm}\label{T:main_result}
There exists a PCE $(\filt^m,S)$. The market signals $\cbra{H_n}$ are from \eqref{E:H_def}.  The filtration $\filt^m$ is constructed as in \eqref{E:filtnk_def}, \eqref{E:filtm_ik_def}. Furthermore,
\begin{enumerate}[(1)]
\item (Prices) Over $[t_n,t_{n+1})$ ($n<N$) or $[t_N,1]$ ($n=N$) the price process takes the form
\begin{equation*}
    S_t = M^{n,X}_t X_t + \sum_{m=1}^n M^{n,m,H}_t H_m + V^{n}_t,
\end{equation*}
where $M^{n,X}, \cbra{M^{n,m,H}}_{m=1}^n$ are continuous, deterministic matrix valued functions of time, and $V^n$ is a continuous deterministic time varying vector. $M_t^{n,X}$ is of full-rank for Lebesgue almost every  $t\in [t_n,t_{n+1})$.

\item (Trading Strategies) For $t\in (t_n,t_{n+1}]$ ($n>1$) or $t\in [0,t_1]$ ($n=1$) the optimal trading strategy for agent $k=0,...,n$ is
\begin{equation}\label{E:strats}
\begin{split}
\gamma_k \wh{\pi}^k_t &= 1_{k>0}\left(C_kG_k- E_k H_k - (C_k-E_k)S_{(t)-}\right)\\
&\qquad  + \ol{\gamma}_n\left(\Pi + \sum_{j=1}^n \alpha_j(C_j-E_j)\left(S_{(t)-}-H_j\right)\right).
\end{split}
\end{equation}

\item (Signal Transmission)  $(S,\filt^m)$ is a DNREE.

\end{enumerate}

\end{thm}

\subsection*{Discussion}

The DNREE property is easiest to see using \eqref{E:resid_demand}. Indeed, at $t_1 = 0$
\begin{equation*}
\wh{\pi}^R_{t_1} = -\alpha_1\alpha_0\ol{\gamma}_1\left(C_1-E_1\right)(S_{t_1}-H_1) + \alpha_1\ol{\gamma}_1\Pi.
\end{equation*}
Thus, if one sees $S_{t_1}$ and $\wh{\pi}^R_{t_1}$ then one can recover $H_1$. Similarly, taking $t\downarrow t_2$
\begin{equation*}
\wh{\pi}^R_{(t_2)+} = -\ol{\gamma}_2\alpha_0\sum_{k=1}^2\alpha_k\left(C_k-E_k\right)(S_{t_2}-H_k) + \ol{\gamma}_2(\alpha_1+\alpha_2)\Pi.
\end{equation*}
As $H_1$ has already been revealed, $H_2$ is revealed at $(t_2)+$  by viewing $S_{t_2}, \wh{\pi}^R_{(t_2)+}$. Continuing in this manner yields the DNREE property, and indicates an important fact. Namely, that even if the public signal $H_k$ is not transferred to the market at the exact same time as the insider signal (i.e. at $t_k$), it is transferred, through the change in noise trader demand and the trading actions of the newly arrived $k^{th}$ insider, at an infinitesimally small time later.

Next, re-write the trading strategy of insider $k=1,...,n$ over $(t_n,t_{n+1}]$ as
\begin{equation*}
\begin{split}
\gamma_k\wh{\pi}^k_t &= \ol{\gamma}_n\Pi + C_k(G_k-S_{t-}) + \left(E_k + \ol{\gamma}_n\alpha_k(C_k-E_k)\right)(S_{t-}-H_k)\\
&\qquad  + \sum_{j=1,j\neq k}^n\alpha_j(C_j-E_j)(S_{t-}-H_j).
\end{split}
\end{equation*}
As $0<\ol{\gamma}_n\alpha_k < 1$ the matrix in front of $S_{t-}-H_k$ is in between (in the sense of positive definite matrices) $E_k$ and $C_k$.  Therefore, there are four components to the optimal policy.  The first ($\ol{\gamma}_n \Pi$) is the optimal trading strategy in the absence of any private information.  The second compares $S_{t-}$ to the private signal $G_k$ (a signal about $S_1=X_1$) and the insider trades positively with respect to the signal. In other words, if the signal $G_k$ exceeds $S_{t-}$ the insider is inclined to think prices will rise, and increases her position in the asset. Else, she thinks the price will fall and takes a negative position.  The third and fourth terms offset the ``optimism'' of the second term, by comparing the price to the public signals, and trading in the negatively.

We can also isolate the effects of private information and noise, by noting from \eqref{E:H_def} that if the $k^{th}$ insider sees $G_k$ and $H_k$, then she sees the noise term $Z_k$, and
\begin{equation*}
\gamma_k\wh{\pi}^k_t = \gamma_0\wh{\pi}^0_t + (C_k-E_k)(G_k-S_{t-}) -\frac{1}{\alpha_k}E_kC_k^{-1}Z_k.
\end{equation*}
Here, the offsetting effects of the private signal and noise are clear. The informed agent acts like the uninformed agent except regarding her private signal and noise.  She trades positively  with respect to her own signal, but negatively with respect to noise.  And, in the limit of vanishing noise trader variance ($E_k\to C_k$, $Z_k \to 0$) the insider's strategy coincides with the uninformed's. Qualitatively this is clear: if $G_k$ is passed to the market, the $k^{th}$ insider is no more well-informed than the uninformed agent.

To conclude we follow up on the ``three equilibrium problem'' comment in Section \ref{SS:two}. Namely, due of the jump in the price process, we may view $[0,1]$ as being comprised of $2N-1$ ``periods'', with $N$ intervals of continuous trading, and $N-1$ jumps, each of which can be associated with a static one period model.  Accordingly, there are two types of equilibria to establish: one in continuous time and one over a single trading interval.  Using CARA preferences we may work backwards, with future optimal trading yielding an associated random endowment.   As such, the proof of Theorem \ref{T:main_result}, carried out in Section \ref{S:pce_construct} below, follows an inductive argument. This argument relies on two ``core'' equilibrium results (one dynamic, one static) in the CARA-Gaussian setting with common information, which we present first.

\section{Exponential-Gaussian Equilibrium with Common Information}\label{S:ExpCaraEq}

Here, we provide equilibrium results in a CARA-Gaussian setting with common information. These results are known in the literature: see \cite{grossman1980impossibility} for the static case, and the dynamic case can be obtained using standard results on complete market equilibria as found in \cite{MR1640352}, but we provide proofs in Appendix \ref{AS:ExpCaraEq}. We will then use these results to establish Theorem \ref{T:main_result}.

For the sake of generality, assume there are $J+1$ agents, and the outstanding supply is $\Psi$. For $j=0,1,...,J$, agent $j$ has CARA preferences with risk aversion $\gamma_j$, weight in the economy $\omega_j$, and weighted risk tolerance $\alpha_j = \omega_j / \gamma_j$. The representative risk aversion is $\ol{\gamma} = (\sum_{j=1}^J \alpha_j)^{-1}$.

\subsection{Continuous Time}\label{SS:cont}  Here, each agent uses the filtration $\filt^B$ and the time interval is $[a,b]$.  The terminal risky asset price is $S_b = \pvct + \pmtx X_b$ for a vector $\pvct$ and matrix $\pmtx$. Agent $j$ has risk adjusted random endowment (notice the presence of both $S_b$ and $X_b$)
\begin{equation}\label{E:cont_rand_endow}
    \gamma_j \E^j_b = 1_{j>0}\left(\frac{1}{2}S_b' M_j S_b - S_b' V_j\right) + \frac{1}{2}X_b'M X_b - X_b'V + \lambda_j,
\end{equation}
where $M_j, M\in\mathbb{S}^d$ and $V_j, V \in\reals^d$ and $\lambda_j\in\reals$.\footnote{Below, when we use Proposition \ref{P:abstract_cont_time}, $\lambda_j$ may depend upon the signals $\cbra{h_n}$. This is why we include it here, even though it plays no role in establishing equilibrium.} We let $Z$ be a to-be-determined strictly positive martingale such that $\expvs{\left(1\vee |X_b|^2\right)Z_b} < \infty$ and define the price process by
\begin{equation}\label{E:cont_price_def}
    S_t \dfn \frac{\condexpvs{Z_{b}S_b}{\F^B_t}}{\condexpvs{Z_{b}}{\F^B_t}};\qquad t\in [a,b].
\end{equation}
The set of acceptable strategies $\A_{a,b}$, common for all agents, consists of those $\filt^B$ predictable, $S$ integrable strategies such that $\condexpv{\qprob}{}{\We^{\pi}_{a,b}}{\F^B_a}\leq 0$ for the measure $\qprob$ associated to $Z$. Agent $j$ has optimal investment problem
\begin{equation*}
    \inf_{\pi\in\A_{a,b}} \condexpvs{e^{-\gamma_j \left(\We^{\pi}_{a,b} + \E^j_b\right)}}{\F^B_a}.
\end{equation*}
We say $Z$ determines a ``complete market equilibrium'' if there exist optimal $\cbra{\wh{\pi}^j}$ such that $\prob\times\Leb_{[a,b]}$ almost surely $\Psi = \sum_{j=0}^J \omega_j \wh{\pi}^j$, and if the $(S,\filt^B)$ market is complete over $[a,b]$.  The following proposition identifies $Z$ enforcing a complete market equilibrium. To state it, define the matrix and vector
\begin{equation}\label{E:cont_olmv}
    \ol{M} \dfn \ol{\gamma}\pmtx'\left(\sum_{j=1}^J \alpha_j M_j\right)\pmtx;\qquad \ol{V} \dfn \ol{\gamma}\pmtx'\left(\sum_{j=1}^J \alpha_j V_j - \Psi - \left(\sum_{j=1}^J \alpha_j M_j\right)\pvct\right),
\end{equation}
and assume
\begin{ass}\label{A:cont_time_verify}
For $P$ in \eqref{E:X_trans_density} and $M$ in \eqref{E:cont_rand_endow}, $P(b-a) + M +  \ol{M} \in \sdpos$.
\end{ass}

\begin{prop}\label{P:abstract_cont_time}
There is a complete market equilibrium with martingale measure
\begin{equation}\label{E:abstract_cont_time_z}
    \frac{d\qprob}{d\prob}\big|_{\F^B_b} = Z_{b};\qquad Z_{b} \dfn \frac{e^{ -\frac{1}{2}X_b'\left(M+\ol{M}\right) X_b + X_b'\left(V+\ol{V}\right)}}{\expvs{e^{- \frac{1}{2}X_b'\left(M+\ol{M}\right) X_b + X_b'\left(V+\ol{V}\right)}}}.
\end{equation}
The price process is
\begin{equation}\label{E:abstract_cont_time_px}
    S_{t} = \pvct + \pmtx\left(P(b-t) + M + \ol{M}\right)^{-1}\left(P(b-t)\mu(X_t,b-t) + V + \ol{V}\right);\qquad t\in [a,b].
\end{equation}
For $j=0,...,J$ the optimal strategy for agent $j$ is
\begin{equation*}
\gamma_j \wh{\pi}^j_t = 1_{j>0}\left(V_j - M_j S_t\right) + \ol{\gamma}\left(\Psi + \sum_{k=1}^J \alpha_k\left(M_k S_t - V_k\right)\right),
\end{equation*}
and the time $a$ certainty equivalent is
\begin{equation*}
    \begin{split}
        \gamma_j \check{\E}^j_a &\dfn -\log\left(\condexpvs{e^{-\gamma_j \left(\wh{\We}^j_{a,b} + \E^j_b\right)}}{\F^B_a}\right) = \lambda_j + \check{\lambda}_j + 1_{j>0}\left(\frac{1}{2}S_a'M_j S_a - S_a'V_j\right)\\
         &\qquad + \check{\E}_a\left(\mu(X_a,b-a),V,\ol{V}\right).
    \end{split}
\end{equation*}
The constant $\check{\lambda}_j$ is explicitly identified in \eqref{E:cont_lambda_def} below, and does not depend upon the vectors $V,\cbra{V_j}_{j=1}^J,\Psi$.  The function $\check{\E}_a(y,v,\ol{v})$ is linear-quadratic in $(y,v,\ol{v})$ with coefficients explicitly identified in \eqref{E:cont_lq_coeff} below. Additionally the Hessian with respect to $y$ is strictly positive definite.
\end{prop}

\begin{rem}\label{R:cont_time}

As $\mu(x,b-a)$ is affine in $x$, $\check{\E}_a$ is linear-quadratic in $(x,V,\ol{V},\textrm{Const})$ and strictly convex in $x$. Also,  the time $a$ certainty equivalent is very similar in structure to the time $b$ random endowment of \eqref{E:cont_rand_endow}, especially with respect to the insider positions in the risky asset. Also, when we use this  result below, the $\{V_j\}, V$ will be affine in the relevant signal realizations $\cbra{g_j},\cbra{h_j}$.   Therefore, prices are jointly linear $x,\cbra{h_j}$, and certainty equivalents are jointly linear quadratic in $x,\cbra{g_j}, \cbra{h_j}$.

\end{rem}

\subsection{Single Period}\label{SS:sp}  We now provide an analogous single period result.  There is a random variable $H\sim N(0,(P^H)^{-1})$. The risky asset terminal price is $S = \pvct + \pmtx H$ for a vector $\pvct$ and matrix $\pmtx$. Agent $j$ has risk adjusted random endowment
\begin{equation*}
    \gamma_j \E^j = 1_{j>0}\left(\frac{1}{2}S' M_j S - S' V_j\right) + \frac{1}{2}H'M H - H'V + \lambda_j,
\end{equation*}
where $M_j, M\in\mathbb{S}^d$, $V_j, V \in\reals^d$ and $\lambda_j\in\reals$. We assume
\begin{ass}\label{A:lq_sp}
$1_{j>0}\pmtx'M_j\pmtx + M + P^H \in\sdpos$ for $j=0,....,J$.
\end{ass}
For a to-be-determined deterministic price $p$, agent $j$ has optimal investment problem
\begin{equation*}
    \inf_{\pi\in\reals^d} \expvs{e^{-\gamma_j \pi'(Y-p) -\gamma_j\E^j}},
\end{equation*}
and we say a single period equilibrium holds for $p$ if there exist optimal policies $\cbra{\wh{\pi}^j}$ such that $\Psi = \sum_{j=0}^J \omega_j \wh{\pi}^j$. To identify $p$, as in the continuous time case define
\begin{equation*}
    \ol{M} \dfn \ol{\gamma}\pmtx'\left(\sum_{j=1}^J \alpha_j M_j\right);\qquad \ol{V} \dfn \ol{\gamma}\pmtx'\left(\sum_{j=1}^J \alpha_j V_j - \Psi - \left(\sum_{j=1}^J \alpha_j M_j\right)\pvct\right).
\end{equation*}

\begin{prop}\label{P:abstract_sp}
There is a single period equilibrium with price
\begin{equation*}
    p = \pvct + \pmtx\left(P^H + M + \ol{M}\right)^{-1}\left(V + \ol{V}\right).
\end{equation*}
For $j=0,...,J$ the optimal strategy for agent $j$ is
\begin{equation*}
\gamma_j \wh{\pi}^j = 1_{j>0}\left(V_j - M_j p\right) + \ol{\gamma}\left(\Psi + \sum_{k=1}^J \alpha_k\left(M_k p - V_k\right)\right),
\end{equation*}
and the certainty equivalent is
\begin{equation*}
    \begin{split}
        \gamma_j \E^{j,-} &\dfn -\log\left(\expvs{e^{-\gamma_j (\wh{\pi}^j)'(Y-p) -\gamma_j \E^j}}\right) = \lambda_j + \check{\lambda_j} +  1_{j>0}\left(\frac{1}{2}p'M_j p - p'V_j\right)+ \E^{-}(V,\ol{V}).
    \end{split}
\end{equation*}
The constant $\check{\lambda}_j$ is explicitly identified in \eqref{E:sp_lambda_def} below, and does not depend upon the vectors $V,\cbra{V_j}_{j=1}^J,\Psi$. The function $\E^{-}(v,\ol{v})$ is linear-quadratic in $(v,\ol{v})$ with coefficients explicitly identified in \eqref{E:sp_lq_coeff} below.
\end{prop}

\begin{rem}\label{R:lin_quad_1}
Note the similarity between the prices, optimal strategies, and endowments in Propositions \ref{P:abstract_cont_time} and \ref{P:abstract_sp}. Also, as in continuous time, provided the initial certainty equivalent is linear-quadratic in the factor process and the signals, the post-equilibrium certainty equivalent is also linear quadratic in the factor process and (remaining) signals.
\end{rem}

\section{Conditional Densities and Connections with Signal Realizations}\label{S:EqEst}

Recall the market signals in \eqref{E:H_def} and the filtrations  in \eqref{E:filtnk_def}. To use the common information equilibrium results in Section \ref{S:ExpCaraEq} we must connect conditional expectations with respect to $\F^{n,k}_{\cdot}$ to those with respect to $\F^B_{\cdot}$.  The purpose of this section is to introduce notation and provide the necessary formulas for making the connection.  This section is technical, but necessary for what follows.  The proof of Lemma \ref{L:Lambda_ident} is given in Appendix \ref{AS:EqEst}.

Fix $t\in [0,1]$, recall \eqref{E:X_trans_density}, and for $M\in\symmat^d$ such that $P(1-t)+M\in\sdpos$, $V\in\reals^d$ and $x\in\reals^d$ define (here, $\tau = 1-t$)
\begin{equation}\label{E:Lambda_def}
    \begin{split}
        \Lambda(t,x;M,V) \dfn& \log\left(\condexpvs{e^{-\frac{1}{2}X_{1}'M X_{1} + X_{1}'V}}{X_t = x}\right),\\
        =& -\frac{1}{2}\log\left(\left| 1_d + \Sigma(\tau)M\right|\right) - \frac{1}{2}\mu(x,\tau)'P(\tau)\mu(x,\tau)\\
        &\qquad\qquad + \frac{1}{2}\left(P(\tau)\mu(x,\tau) + V\right)'\left(P(\tau)+M\right)^{-1} \left(P(\tau)\mu(x,\tau) + V\right).
    \end{split}
\end{equation}
Next, for $n\leq N$ define
\begin{equation}\label{E:N_matrices_def}
\ol{h}_n \dfn (h_1,...,h_n);\qquad M^n \dfn \sum_{m=1}^n E_n;\qquad V^n(\ol{h}_n) \dfn \sum_{m=1}^n E_n h_n.
\end{equation}
as well as the functions
\begin{equation}\label{E:ell_def_0}
\ell_n(t,x,\ol{h}_n) \dfn \Lambda\left(t,x; M^n, V^n(\ol{h}_n)\right);\qquad p^{n,0}(t,x,\ol{h}_n) \dfn \frac{e^{\ell_n(t,x,\ol{h}_n)}}{e^{\ell_n(0,X_0,\ol{h}_n)}}.
\end{equation}
Similarly, for $1\leq k\leq n$ define
\begin{equation}\label{E:ell_def_k}
\begin{split}
\ol{h}_n^{-k} &\dfn (h_1,...,h_{k-1},h_{k+1},...,h_n),\\
\ell_{n,k}(t,x,\ol{h}^{-k}_n,g_k) &\dfn \Lambda\left(t,x;C_k-E_k + M^n,  C_k g_k - E_k h_k  + V(\ol{h}_n)\right),\\
p^{n,k}(t,x,\ol{h}^{-k}_n,g_k) &\dfn \frac{e^{\ell_{n,k}\left(t,x,\ol{h}^{-k}_n,g_k\right)}}{e^{\ell_{n,k}\left(0,X_0,\ol{h}_n^{-k},g_k\right)}}.
\end{split}
\end{equation}
With this notation, we verify the Jacod  conditions in the following lemma.

\begin{lem}\label{L:Lambda_ident} For $0\leq t \leq 1$ and $n\leq N$
\begin{equation}\label{E:Lambda_ident}
\begin{split}
\frac{\condprobs{\ol{H}_n \in d\ol{h}_n}{\F^B_t}}{\prob\bra{\ol{H}_n\in d\ol{h}_n}} &= p^{n,0}(t,X_t,\ol{h}_n).
\end{split}
\end{equation}
For $1\leq k\leq n$
\begin{equation}\label{E:Lambda_ident_2}
\begin{split}
&\frac{\condprobs{G_k \in dg_k}{\F^B_t\vee\sigma(\ol{H}_n)}}{\condprobs{G_k \in dg_k}{\sigma(\ol{H}_n)}}  = \frac{p^{n,k}(t,X_t,\ol{H}^{-k}_n,g_k)}{p^{n,0}(t,X_t,\ol{H}_n)}.
\end{split}
\end{equation}
\end{lem}

Lemma \ref{L:Lambda_ident} implies the following result, crucial for connecting conditional expectations with respect to $\F^{n,k}_{\cdot}$ with those with respect to $\F^B_{\cdot}$. Namely, for $n\leq N, 1\leq k\leq n$, $s < t < 1$ and all generic functions $\chi$,  on the set $\cbra{\ol{H}_n = \ol{h}_n, G_k = g_k}$\footnote{For two random variables $(X,Y)$, if  $\condprobs{Y\in dy}{X} = p^y(X)\times \prob\bra{Y\in dy}$ then $\condexpvs{\chi(X)}{Y} = m(Y)$ where $m(y) = \expvs{p^y(X)\chi(X)}$. Thus, $\condexpvs{\chi(X)}{Y} = m(y)$ on $\cbra{Y=y}$, and the identity is valid even if $\prob\bra{Y=y} = 0$ for all $y$.}
\begin{equation}\label{E:big_condexp_ident}
    \begin{split}
        \condexpvs{\chi(X_t, \ol{H}_n)}{\F^{n,0}_s} &= \condexpvs{\frac{p^{n,0}(t,X_t,\ol{h}_n)}{p^{n,0}(s,X_s,\ol{h}_n)} \chi(X_t,\ol{h}_n)}{\F^B_s},\\
        \condexpvs{\chi(X_t, \ol{H}_n, G_k)}{\F^{n,k}_s} &= \condexpvs{\frac{p^{n,k}(t,X_t,\ol{h}^{-k}_n,g_k)}{p^{n,k}(s,X_s,\ol{h}^{-k}_n,g_k)} \chi(X_t,\ol{h}_n,g_k)}{\F^B_s}.
    \end{split}
\end{equation}
These identities allow us to pass to the signal-realization level, establish equilibrium using Propositions \ref{P:abstract_cont_time}, \ref{P:abstract_sp}, and then lift results back up to the signal level. In the next section we carry out this plan. While the last period requires a slightly different argument, we can establish equilibrium over the jumps and all the remaining periods using a recursive argument.  As such, we start off treating the last period, and then identify the backwards recursive argument.


\section{Construction of the PCE}\label{S:pce_construct}

\subsection*{The last period $[t_{N},1]$} Here, all agents participate, and $k\leq N$ is fixed. The proof of Proposition \ref{P:eq_last_period} is in Appendix \ref{AS:last_period}. We start motivating the signal-realization equilibrium problem, recalling \eqref{E:N_matrices_def}. The uninformed agent uses $\filt^{N,0}$, and for generic functions $\chi$, \eqref{E:ell_def_0} and \eqref{E:big_condexp_ident} imply on $\cbra{\ol{H}_N = \ol{h}_N}$
\begin{equation}\label{E:U_ident_1}
    \begin{split}
        \condexpvs{\chi(X_1,\ol{H}_N)}{\F^{N,0}_{t_N}} &=e^{-\ell_N(t_N,X_{t_N},\ol{h}_N)}\condexpvs{\chi(X_1,\ol{h}_N)e^{-\frac{1}{2}X_1' M^{N} X_1 + X_1'V^N(\ol{h}_N)}}{\F^{B}_{t_N}}.
    \end{split}
\end{equation}
The $k^{th}$ insider uses $\filt^{N,k}$ and \eqref{E:ell_def_k}, \eqref{E:big_condexp_ident} imply on $\cbra{\ol{H}_N = \ol{h}_N, G_k = g_k}$
\begin{equation}\label{E:Ik_ident_1}
    \begin{split}
        &\condexpvs{\chi(X_1,\ol{H}_N,G_k)}{\F^{N,k}_{t_N}} = e^{-\ell_{N,k}(t_N,X_{t_N},\ol{h}^{-k}_N,g_k)}\\
        &\qquad \times \condexpvs{\chi(X_1,\ol{h}_N,g_k)e^{-\frac{1}{2}X_1' M^{N} X_1 + X_1'V^N(\ol{h}_N) -\frac{1}{2} X_1'(C_k-E_k)X_1 + X_1'\left(C_k g_k -E_k h_k\right)}}{\F^{B}_{t_N}}.
    \end{split}
\end{equation}

\subsubsection*{Signal Realization Level Equilibrium} \eqref{E:U_ident_1} and \eqref{E:Ik_ident_1} suggest  that for fixed $\ol{h}_N, \cbra{g_k}$ (hence fixed noise trader demands $\cbra{z_k}$) we should consider the ``signal-realization equilibrium'' problem over $[t_N,1]$ where all agents have common information flow $\filt^B$ and agent $k$ has random endowment
\begin{equation*}
\gamma_k \E_k =  1_{k>0}\left(\frac{1}{2} X_1'(C_k-E_k)X_1 - X_1'\left(C_k g_k -E_k h_k\right)\right) + \frac{1}{2}X_1' M^{N} X_1 - X_1'V^N(\ol{h}_N).
\end{equation*}
This is exactly the setting of Section \ref{SS:cont} for appropriate $\pvct,\pmtx,\cbra{V_k},\cbra{M_k}$ and accounting for noise trading, outstanding supply $\Psi = \Pi - \sum_{m=1}^N z_m$.  As
\begin{equation*}
P(1-t_N) + M + \ol{M} = P(1-t_N) + \sum_{m=1}^N \left(E_m + \ol{\gamma}_N \alpha_m (C_m-E_m)\right),
\end{equation*}
is positive definite, Proposition \ref{P:abstract_cont_time} yields a complete market signal realization equilibrium over $[t_N,1]$.  The price process is $S^{N,\ol{H}_N}_t = S^N(t,X_t,\ol{h}_N)$ where
\begin{equation}\label{E:last_period_price_function}
    S^N(t,x,\ol{h}_N) \dfn \left(P(1-t) + \M_N^{N}\right)^{-1}\left(P(1-t)\mu(x,1-t) + \V_N^{N}(\ol{h}_N) - \ol{\gamma}_N \Pi\right),
\end{equation}
with
\begin{equation*}
    \begin{split}
        \M_n^N &= \sum_{m=1}^n \left(E_m + \ol{\gamma}_N \alpha_m(C_m-E_m)\right);\qquad \V_n^N(\ol{h}_n) =\sum_{m=1}^n \left(E_m + \ol{\gamma}_N \alpha_m(C_m-E_m)\right)h_m .
    \end{split}
\end{equation*}
The martingale measure $\qprob^{N,\ol{h}_N}$ has density process  $Z^{N,\ol{h}_N}_\cdot = Z^N(\cdot,X_\cdot,\ol{h}_N)$ over $[t_N,1]$ where
\begin{equation}\label{E:Nmkt_density}
    Z^N(t,x,\ol{h}_N) = \frac{\condexpvs{e^{-\frac{1}{2}X_1'\M_N^N X_1 + X_1'\left(\V_N^N(\ol{h}_N)-\ol{\gamma}_N\Pi\right)}}{\F^B_t}}{\expvs{e^{-\frac{1}{2}X_1'\M_N^N X_1 + X_1'\left(\ol{V}_N^N(\ol{h}_N)-\ol{\gamma}_N\Pi\right)}}}.
\end{equation}
Lastly, the trading strategies are $\wh{\pi}^k_t = \wh{\pi}^{N,k}\left(S^{N,\ol{h}_N}_t, \left(g_k, \ol{h}_N\right)\right)$, $k=0,...,N$ where
\begin{equation}\label{E:last_cont_strat}
\begin{split}
\wh{\pi}^{N,k}\left(s, \left(g_k,\ol{h}_N\right)\right)&=\frac{1_{k>0}}{\gamma_k}\left(C_k g_k - E_k h_k- (C_k-E_k)s\right)\\
&\qquad\qquad  + \frac{\ol{\gamma}_N}{\gamma_k}\left(\Pi + \sum_{m=1}^N \alpha_m(C_m-E_m)\left(s - h_m\right)\right),
\end{split}
\end{equation}
which is exactly as in Theorem \ref{T:main_result} as $S$ is continuous over $[t_N,1]$.

\subsubsection*{Signal Level Equilibrium} Using the signal realization level results, we now establish equilibrium for the public filtration $\filt^{N,0}$; the price process $t\to S^{N,\ol{H}_N}_t = S^N(t,X_t,\ol{H}_N)$ (which takes the form in Theorem \ref{T:main_result}); and using \eqref{E:last_cont_strat}, where the candidate optimal strategies are
\begin{equation}\label{E:last_cont_strat_sig}
\wh{\pi}^k_{\cdot} = \wh{\pi}^k\left(S^{N,\ol{H}_N}_{\cdot}, \left(1_{k>0}G_k, \ol{H}_N\right)\right).
\end{equation}
Note these strategies clear the market for all $\cbra{g_k}$ and $\ol{h}_N$, hence for $\cbra{G_k},\ol{H}_N$. First, as in the setting of Proposition \ref{P:abstract_cont_time}, for fixed $\ol{h}^N$ define the admissible strategies
\begin{equation}\label{E:lp_admit}
\begin{split}
    \A^{\ol{h}_N}_{t_N,1} &\dfn \bigg\{ \pi\in\mcp(\filt^B) \such \pi \textrm{ is } S^{N,\ol{h}_N} \textrm{ integrable over }[t_N,1],\\
     &\qquad \int_{t_N}^1 \pi_t'dS^{N,\ol{h}_N}_t \in L^1(\qprob^{N,\ol{h}_N}) \textrm{ with } \condexpv{\qprob^{N,\ol{h}_N}}{}{\int_{t_N}^1 \pi_t'dS^{N,\ol{h}_N}_t}{\F^B_{t_N}} \leq 0\bigg\}.
\end{split}
\end{equation}
Above, for a given filtration $\filt$, we have written $\mcp(\filt)$ for the $\filt$ predictable sigma field and we write $\pi\in\mcp(\filt)$ for ``$\pi$ is $\mcp(\filt)$ measurable''.

Next, with $\B(\reals^p)$ denoting the Borel sigma field over $\reals^p$ for an integer $p$, we define the admissible strategies for the uninformed agent (agent $0$) and $k^{th}$ insider (agent $k$) to be
\begin{equation*}
    \begin{split}
        \A^{N,0}_{t_N,1} &\dfn \cbra{\pi \in \mcp(\filt^B)\otimes \B(\reals^{d\times N}) \such \pi^{\ol{h}_N} \in \A^{\ol{h}_N}_{t_N,1} \textrm{ for } \textrm{ Leb a.e. } \ol{h}_N\in\reals^{d\times N} },\\
        \A^{N,k}_{t_N,1} &\dfn \cbra{\pi \in \mcp(\filt^B)\otimes \B(\reals^{d\times (1+N)}) \such \pi^{g_k,\ol{h}_N} \in \A^{\ol{h}_N}_{t_N,1} \textrm{ for } \textrm{ Leb a.e. } (g_k,\ol{h}_N)\in\reals^{d\times (1+N)} },
    \end{split}
\end{equation*}
where we have written a typical $\mcp(\filt^B)\otimes \B(\reals^p)$ measurable map $(t,\omega,y) \to \pi(t,\omega,y)$ as $\pi^{y}$.  In words, a policy is admissible at the signal level if it (is appropriately measurable and ) is admissible at the signal realization level for (almost surely) all possible realizations.

From \eqref{E:last_cont_strat} and \eqref{E:last_cont_strat_sig} it is clear that $\wh{\pi}^k \in \A^{N,k}_{t_N,1}$ for $k=0,...,N$.  Appendix \ref{AS:strat} contains additional information on $\A^{N,k}_{t_N,1}$, and in particular connects these sets to the ``standard'' set of admissible strategies, comprising of $\filt^{N,k}$ predictable, $S^{N,\ol{H}_N}$ integrable policies satisfying a budget constraint with respect to martingale measure $\qprob^{N,k}$ (which is also identified in Appendix \ref{AS:strat}) related to $\qprob^{N,\ol{h}_n}$.  As this connection is both technical, and modulo notation, identical to that presented in \cite{detemple2022dynamic}, we postpone it until Appendix \ref{AS:strat}.

With these preparatory results,  over $[t_N,1]$ the optimal investment problems in \eqref{E:vf_def} specify to
\begin{equation}\label{E:vf_def_lp}
    \begin{split}
        &\inf_{\pi\in \A^{N,0}_{t_N,1}} \condexpvs{e^{-\gamma_0 \We^{\pi}_{t_N,1}}}{\F^m_{t_N}};\qquad \inf_{\pi\in \A^{N,k}_{t_N,1}} \condexpvs{e^{-\gamma_k \We^{\pi}_{t_{N},1}}}{\F^{I_k}_{t_N}},
    \end{split}
\end{equation}
and we claim

\begin{prop}\label{P:eq_last_period}
$(S^N,\filt^{N,0})$ is a PCE, and the $(S^N,\filt^{N,0})$, $\cbra{(S^N,\filt^{N,k}}_{k=1}^N$ markets are complete over $[t_N,1]$.
\end{prop}

Lastly, we compute the effective endowment at $t_N$ implied by optimal trading over $[t_N,1]$ for agents $k\leq N$, with an eye towards Proposition \ref{P:abstract_sp}. To this end, write $\wh{\We}^k$ as the optimal gains process for agent $k$ and define
\begin{equation*}
\begin{split}
\wt{\E}^{k}_{t_N} &\dfn  -\log\left(\condexpvs{e^{-\gamma_{k} \wh{\We}^{k}_{t_N,T}}}{\F^{N,k}_{t_N}}\right).
\end{split}
\end{equation*}
For the uninformed agent we have, using Proposition \ref{P:abstract_cont_time}, \eqref{E:U_ident_1} and well as  \eqref{E:intuit_meas} and \eqref{E:intuit_relation} below
\begin{equation}\label{E:checkE_tN_def_U}
    \begin{split}
    \wt{\E}^0_{t_N} &= \ell_N(t_N,X_{t_N},\ol{H}_N) + \check{\lambda}^N_0+ \check{\E}^0_N(X_{t_N},\ol{H}_N),
    \end{split}
\end{equation}
where $\check{\E}^0_N(x,\ol{h}_N)$  (see $\E_a$ from Proposition \ref{P:abstract_cont_time}) is linear-quadratic in $(x,\ol{h}_N)$, and takes the form in Proposition \ref{P:abstract_sp}\footnote{In Proposition \ref{P:abstract_sp}, $H=H_N$. Also, note that  $\check{\E}^0_N(X_{t_N},\ol{H}_N)$ has a component, embedded into $\lambda_j$, that is linear-quadratic in $(X_{t_N},\ol{H}_{N-1})$.}. Similarly, for the $k^{th}$ insider
\begin{equation}\label{E:checkE_tN_def}
\begin{split}
\wt{\E}^k_{t_N} &= \ell_{N,k}(t_N,X_{t_N},\ol{H}_N^{-k},G_k) +\check{\lambda}^N_k +\frac{1}{2}(S^{N,\ol{H}_N}_{t_N})'(C_k-E_k)S^{N,\ol{H}_N}_{t_N}\\
&\qquad\qquad  - (S^N(t_N,X_{t_N},\ol{H}_N))'(C_k G_k - E_k H_k) + \check{\E}^0_{t_N}(X_{t_N},\ol{H}_N),
\end{split}
\end{equation}
which is also jointly linear quadratic in $(x,\ol{h}_N)$. Lastly, Assumption \ref{A:lq_sp} holds, as will be shown for all $n=2,...,N$ and $k\leq n-1$, in Lemma \ref{L:sp_ass_verify}.

\begin{rem}\label{R:why_ou}
Equilibrium over $[t_n,1]$ could be established following the arguments of \cite{MR4192554} (suitably altering the noise trader convention there-in), and one may allow $X$ to be a rather general time-homogenous diffusion.  However, to establish equilibrium over the remaining $N-1$ intervals of continuous trading (not to mention the jumps), we exploit the OU structure, as it yields effective endowments that are explicitly replicable (i.e. without appealing to martingale representation). For general processes $X$ this is not the case, and one must assume, then verify, market completeness. Unfortunately, it is very difficult to do this, because known results (c.f. \cite{anderson2008equilibrium, MR3131287, MR3590708}) require non-degeneracy and regularity conditions on the terminal payoff (here $S_{(t_n)-}, n=2,...,N$) and effective random endowments. While in those papers, these values were explicit functions of the factor process, in our setting they are implicit equilibrium quantities.  For this reason, we consider the OU process and leave extensions to general processes for future research.
\end{rem}


\subsection*{Equilibrium over a jump }

We now establish equilibrium over a jump at $t_n, n=2,...,N$. The price immediately after the jump is (the previously obtained) $S^{n,\ol{H}_n}_{t_n} = S^n(t_n,X_{t_n},\ol{H}_n)$, and we assume $S^{n}$ is affine in $X_{t_n}$ and $\ol{H}_n$. Next, we set $p^n=p^n(X_{t_n},\ol{H}_{n-1})$ as a candidate price  immediately before the jump (the left limit of the equilibrium price as $t\uparrow t_n$).  We assume the participating agents $k=0,...,n-1$, have endowments from trading over $[t_n,T]$ consistent with \eqref{E:checkE_tN_def_U}, \eqref{E:checkE_tN_def} in that
\begin{equation}\label{E:checkE_tn_def}
\begin{split}
\wt{\E}^{0}_{t_n} &= \ell_n(t_n,X_{t_n},\ol{H}_n) + \check{\lambda}^n_0 + \check{\E}^0_n(X_{t_n},\ol{H}_n),\\
\wt{\E}^k_{t_n} &= \ell_{n,k}(t_n,X_{t_n},\ol{H}_n^{-k},G_k) +\check{\lambda}^n_k +\frac{1}{2}(S^{n,\ol{H}_n}_{t_n})'(C_k-E_k)S^{n,\ol{H}_n}_{t_n}\\
&\qquad\qquad  - (S^n(t_n,X_{t_n},\ol{H}_n))'(C_k G_k - E_k H_k) + \check{\E}^0_{n}(X_{t_n},\ol{H}_n).
\end{split}
\end{equation}
Furthermore, we assume $\check{\E}^0_{t_n}$ is linear-quadratic in $(x,\ol{h}_n)$ and in preparation for use of Proposition \ref{P:abstract_sp} we write
\begin{equation}\label{E:jump_lambda_form}
\check{\E}^0_n(x,\ol{h}_n) = \lambda_{n}(x,\ol{h}_{n-1}) + \frac{1}{2}h_n'\check{M}_n h_n - h_n'\check{V}_n(x,\ol{h}_{n-1}),
\end{equation}
where $\check{M}_n$ satisfies Assumption \ref{A:lq_sp}. The uninformed agent prior to the jump has information set $\F^{n-1,0}_{t_n}$, and solves
\begin{equation*}
\inf_{\pi\in\F^{n-1,0}_{t_{n}}} \condexpvs{e^{-\gamma_0 \pi'(S^{n,\ol{H}_n}_{t_n}-p^n) - \wt{\E}^0_{t_n}}}{\F^{n-1,0}_{t_{n}}}.
\end{equation*}
The only random variable not known given $\F^{n-1,0}_{t_n}$ is $H_n$, and our goal is to cast this minimization problem, on the set $\cbra{(X_{t_n},\ol{H}_{n-1}) = (x,\ol{h})}$, as one where we take regular expectations with respect to certain function of $x,\ol{h}$ and a random variable $H\sim N(0,E_n^{-1})$. This will enable us to invoke Proposition \ref{P:abstract_sp}.  To do this, we note (using Lemma \ref{L:Lambda_ident}) that $H_n$ has $\F^{n-1,0}_{t_n}$ conditional probability density function (pdf)
\begin{equation*}
\frac{e^{\ell_n(t_n,X_{t_n}, (\ol{H}_{n-1},h_n))}}{e^{\ell_{n-1}(t_n,X_{t_n},\ol{H}_{n-1})}}\times p_{E_n}(h_n).
\end{equation*}
where $p_C$ is the $N(0,C^{-1})$ pdf. Thus, fixing the set $\cbra{(X_{t_n},\ol{H}_{n-1}) = (x,\ol{h})}$, suppressing $(t_n,x,\ol{h})$ from all function arguments, and appealing to \eqref{E:checkE_tn_def}
\begin{equation*}
    \begin{split}
     \condexpvs{e^{-\gamma_0 \pi'(S^{n,\ol{H}_n}_{t_n}-p^n) - \wt{\E}^0_{t_n}}}{\F^{n-1,0}_{t_{n}}} &= \int_{h_n} e^{-\gamma_0 \pi'(S^n(h_n)-p^n) -\check{\lambda}^n_0 -\check{\E}^0_n(h_n) - \ell_{n-1}}p_{E_n}(h_n)dh_n,\\
    &= e^{- \ell_{n-1}-\check{\lambda}^n_0}\expvs{e^{-\gamma_0 \pi'(S^n(H)-p^n)-\check{\E}^0_n(H) }};\quad H\sim N(0,E_n^{-1}).
    \end{split}
\end{equation*}
This leaves the minimization problem
\begin{equation*}
\begin{split}
&\inf_{\pi\in\reals^d} \expvs{e^{-\gamma_0 \pi'(S^n(H)-p^n)-\check{\E}^0_n(H) }}.
\end{split}
\end{equation*}
We now turn to the $k^{th}$ insider, $k=1,...,n-1$ who solves
\begin{equation*}
    \inf_{\pi\in\F^{n-1,k}_{t_n}} \condexpvs{e^{-\gamma_k\pi'(S^n_{t_n}-p^n) - \wt{\E}^{k}_{t_n}}}{\F^{n-1,k}_{t_n}}.
\end{equation*}
Similarly to the uninformed, one can show that given $\F^{n-1,k}_{t_n}$, $H_n$ has the pdf
\begin{equation*}
\frac{e^{\ell_{n,k}(t_n,X_{t_n}, (\ol{H}_{n-1}^{-k},h_n),G_k)}}{e^{\ell_{n-1,k}(t_n,X_{t_n},\ol{H}_{n-1}^{-k},G_k)}} \times p_{E_n}(h_n).
\end{equation*}
Therefore, fixing $\cbra{(X_{t_n},\ol{H}_{n-1},G_k) = (x,\ol{h},g_k)}$\footnote{Note that $k\leq n-1$ implies $\ol{h}_n^{-k} = (\ol{h}_{n-1}^{-k},h_n)$. Also, we will write $h_k$ for the realization of $H_k$.}, suppressing $(t_n,x,\ol{h})$, and appealing to \eqref{E:checkE_tn_def}, we have
\begin{equation*}
    \begin{split}
        &\condexpvs{e^{-\gamma_k \pi'(S^n_{t_n}-p^n) - \wt{\E}^{k}_{t_n}}}{\F^{n-1,k}_{t_n}}\\
        &\quad = \int_{h_n} e^{-\gamma_k \pi'(S^n(h_n)-p^n)-\check{\lambda}^n_k -\frac{1}{2}S^n(h_n)'(C_k-E_k)S^n(h_n) + (C_kg_k-E_kh_k)'S^n(h_n) - \check{\E}^{0}_n(h_n) -\ell_{n-1,k}(g_k)} p_{E_n}(h_n)dh_n,\\
        &\quad = e^{-\ell_{n-1,k}(g_k)-\check{\lambda}^n_k} \expvs{e^{-\gamma_k \pi'(S^n(H)-p^n) -\frac{1}{2}S^n(H)'(C_k-E_k)S^n(H) + (C_kg_k-E_kh_k)'S^n(H) - \check{\E}^{0}_n(H)}},
    \end{split}
\end{equation*}
where again $H\sim N(0,E_n^{-1})$. This leaves the minimization problem
\begin{equation*}
\begin{split}
    \inf_{\pi\in\reals^d} \expvs{e^{-\gamma_k \pi'(S^n(H)-p^n) -\frac{1}{2}S^n(H)'(C_k-E_k)S^n(H) + (C_kg_k-E_kh_k)'S^n(H) - \check{\E}^{0}_n(H)}}.
\end{split}
\end{equation*}
The outstanding supply, accounting for the noise terms is $\Psi = \Pi - \sum_{k=1}^{n-1}z_k$. Given the form for $\check{\E}_n^0$ in \eqref{E:jump_lambda_form} and Lemma \ref{L:sp_ass_verify}, the hypotheses of Proposition \ref{P:abstract_sp} are met, and we obtain
\begin{prop}\label{P:jump}
    Equilibrium holds over the jump for $\wh{p}^n(X_{t_n},\ol{H}_{n-1})$ where $\wh{p}^n(x,\ol{h})$ is from Proposition \ref{P:abstract_sp} for the appropriate model parameters.
\end{prop}

Similarly to \eqref{E:last_cont_strat}, \eqref{E:last_cont_strat_sig}, the optimal trading strategies over the jump, as computed in Proposition \ref{P:abstract_sp} take the form
\begin{equation*}
\begin{split}
\gamma_k \wh{\pi}^k &= 1_{k>0}\left(C_k G_k - E_k H_k - (C_k-E_k)\wh{p}^n\left(X_{t_n},\ol{H}_{n-1}\right)\right)\\
&\qquad + \ol{\gamma}_{n-1}\left(\Pi + \sum_{m=1}^{n-1} \alpha_m(C_m-E_m)\left(\wh{p}^n\left(X_{t_n},\ol{H}_{n-1}\right)-H_m\right)\right).
\end{split}
\end{equation*}
As $\wh{p}^n(X_{t_n},\ol{H}_{n-1}) = S_{(t_n)-}$ we obtain the formula in Theorem \ref{T:main_result}.  Also this formula coincides with the continuous trading, over $[t_{n-1},t_n]$, optimal strategy at $t_n$.


We finish computing the effective endowments immediately prior to the jump.  First, define
\begin{equation*}
\wt{\E}^{k,-}_{t_n} \dfn -\log\left(\condexpvs{e^{-\gamma_k\left(\wh{\pi}^k\right)'(S^n_{t_n}-\wh{p}^{n}) - \wt{\E}^k_{t_n}}}{\F^{n-1,k}_{t_n}}\right);\qquad k = 0,...,n-1.
\end{equation*}
Next, for the uninformed agent, define $\check{\E}^{0,-}_{n}$ through
\begin{equation*}
\begin{split}
\wt{\E}^{0,-}_{t_n} & = \ell_{n-1}(t_n,X_{t_n},\ol{H}_{n-1}) + \check{\lambda}^{n,-}_0 + \check{\E}^{0,-}_{n}(X_{t_n},\ol{H}_{n-1}),
\end{split}
\end{equation*}
so that $\check{\E}^{0,-}_n$ is the quantity from Proposition \ref{P:abstract_sp}, and hence $\check{\E}^{0,-}_n$ is jointly linear-quadratic in $(x,\ol{h}_{n-1})$. For the $k^{th}$ insider, $k=1,...,n-1$ according Proposition \ref{P:abstract_sp} we have
\begin{equation}\label{E:Ik_prejump_n_endow}
\begin{split}
\wt{\E}^{k,-}_{t_n} & = \ell_{n-1,k}(t_n,X_{t_n},\ol{H}^{-k}_{n-1},G_k)  +\check{\lambda}^{n,-}_k +\frac{1}{2}\wh{p}^{n}(X_{t_n},\ol{H}_{n-1})'(C_k-E_k)\wh{p}^{n}(X_{t_n},\ol{H}_{n-1})\\
&\qquad  - (C_k G_k - E_k H_k)'\wh{p}^{n}(X_{t_n},\ol{H}_{n-1})+ \check{\E}^{0,-}_{n}(X_{t_n},\ol{H}_{n-1}).
\end{split}
\end{equation}
Lastly, Assumption \ref{A:cont_time_verify} holds, as shown in Lemma \ref{L:sp_ass_verify}.


\subsection*{The period $[t_n,t_{n+1}]$ trading in $S^c$ for $n=1,...,N-1$}

We lastly consider when the participating agents $0,...,n$ trade in the continuous part of $S$ over $[t_{n},t_{n+1}]$ for $n=1,....,N-1$.  To ease the notational burden we set $a_n = t_n$,  $b_n = t_{n+1}$ and assume $k\leq n$.  The price at $(b_{n})-$ is the (previously established) $\wh{p}^{n+1}(X_{b_{n}},\ol{H}_{n})$ and is jointly linear in $(x,\ol{h}_n)$. The agents have time $(b_n)-$ endowments
\begin{equation*}
\begin{split}
\wt{\E}^{0,-}_{b_n} &= \ell_{n}(b_{n},X_{b_{n}},\ol{H}_{n}) +\check{\lambda}^{n,-}_0 +\check{\E}^{0,-}_{n+1}(X_{b_{n}},\ol{H}_{n}),\\
\wt{\E}^{k,-}_{b_n} &= \ell_{n,k}(b_{n},X_{b_{n}},\ol{H}^{-k}_{n},G_k) +\check{\lambda}^{n,-}_k + \frac{1}{2}\wh{p}^{n+1}(X_{b_{n}},\ol{H}_{n})'(C_k-E_k)\wh{p}^{n+1}(X_{b_{n}},\ol{H}_{n})\\
&\qquad   - (C_k G_k - E_k H_k)'\wh{p}^{n+1}(X_{b_{n}},\ol{H}_{n}) + \check{\E}^{0,-}_{n+1}(X_{b_{n}},\ol{H}_{n}),
\end{split}
\end{equation*}
which are linear-quadratic in $(x,\ol{h}_n)$ and such that Assumption \ref{A:cont_time_verify} holds for $\check{\E}^{0,-}_{n+1}$ and $\cbra{C_k-E_k}$. The analysis below is very similar to that over $[t_N,1]$ and hence we do not provide all the details, especially with regards to the acceptable trading sets, and lifting results from the signal realization level to the signal level.

We start with the uninformed agent. Using \eqref{E:ell_def_0} and \eqref{E:big_condexp_ident} we conclude that on $\cbra{\ol{H}_n = \ol{h}_n}$
\begin{equation*}
    \condexpvs{\chi(X_{b_{n}},\ol{H}_n)e^{-\wt{\E}^{0,-}_{b_{n}}}}{\F^{n,0}_{a_n}} = e^{-\ell_n(a_n,X_{a_n},\ol{h}_n)-\check{\lambda}^{n,-}_0}\condexpvs{\chi(X_{b_{n}},\ol{h}_n)e^{-\check{\E}^{0,-}_{n+1}(X_{b_{n}},\ol{h}_n)}}{\F^B_{a_n}}.
\end{equation*}
Therefore, the uninformed agent has signal realization optimal investment problem
\begin{equation*}
    \inf_{\pi \in \wt{\A}_{a_n,b_{n}}} \condexpvs{e^{-\gamma_0 \We^{\pi}_{a_n,b_{n}} -\check{\E}^{0,-}_{n+1}(X_{b_{n}},\ol{h}_n)}}{\F^B_{a_n}},
\end{equation*}
where the acceptable class $\wt{\A}_{a_n,b_{n}}$ is determined by the martingale measure $\qprob^{n,\ol{h}_n}$ and associated density process $Z^{n,\ol{h}_n}$ (both of which Proposition \ref{P:abstract_cont_time} will identify). For the $k^{th}$ insider, using \eqref{E:ell_def_k} and \eqref{E:big_condexp_ident} we conclude that on $\cbra{\ol{H}_n =\ol{h}_n,G_k = g_k}$
\begin{equation*}
    \begin{split}
    &\condexpvs{\chi(X_{b_{n}},\ol{H}_N,G_k)e^{-\wt{\E}^{0,k}_{b_{n}}}}{\F^{n,k}_{a_n}} = e^{-\ell_{n,k}(a_n,X_{a_n},\ol{h}^{-k}_n,g_k)-\check{\lambda}^{n,-}_k}\\
    &\qquad \times \mathbb{E}\bigg[\chi(X_{b_{n}},\ol{h}_n,g_k)e^{-\frac{1}{2} (\wh{p}^{n+1}(X_{b_{n}},\ol{h}_{n}))'(C_k-E_k)\wh{p}^{n+1}(X_{b_{n}},\ol{h}_{n})}\\
     &\qquad\qquad \times e^{(C_k g_k - E_k h_k)'\wh{p}^{n+1}(X_{b_{n}},\ol{h}_{n})-\check{\E}^{0,-}_{n+1}(X_{b_{n}},\ol{h}_n)}\ \bigg|\ \F^B_{a_n}\bigg].
    \end{split}
\end{equation*}
This yields signal realization optimal investment problem
\begin{equation*}
    \begin{split}
    \inf_{\pi \in \wt{\A}_{a_n,b_{n}}}& \bigg[e^{-\gamma_k \We^{\pi}_{a_n,b_{n}} - \frac{1}{2} (\wh{p}^{n+1}(X_{b_{n}},\ol{h}_{n}))'(C_k-E_k)\wh{p}^{n+1}(X_{b_{n}},\ol{h}_{n})}\\
     &\qquad\qquad \times e^{(C_k g_k - E_k h_k)'\wh{p}^{n+1}(X_{b_{n}},\ol{h}_{n})-\check{\E}^{0,-}_{n+1}(b_{n},X_{b_{n}},\ol{h}_n)}\ \bigg| \ \F^B_{a_n}\bigg].
    \end{split}
\end{equation*}
The adjusted outstanding supply, accounting for the noise term is $\Psi = \Pi - \sum_{k=1}^n z_k$. As the terminal price $\wh{p}^{n+1}$ is affine in $X_{t+1}$ and the modified endowments $\cbra{\check{\E}^{k,-}_{n+1}}_{k=0}^{n}$ are linear-quadratic in $(X_{t+1},\ol{H}_n)$, and Assumption \ref{A:cont_time_verify} holds, we may invoke Proposition \ref{P:abstract_cont_time} to obtain the signal realization equilibrium, and associated price process $t\to S^{n,\ol{h}_n}(t,X_t,\ol{h}_n)$ and trading strategies $\cbra{\wh{\pi}^k}$.  We then obtain the signal level equilibrium using the exact same arguments as over $[t_N,1]$, and write $S^{n,c}_t = S^{n,\ol{H}_n}(t,X_t,\ol{H}_n)$.  This yields

\begin{prop}\label{P:eq_cont}
$(S^{n,c},\filt^{n,0})$ form a PCE over $[t_n,t_{n+1}]$. The optimal trading strategies take the form in Theorem \ref{T:main_result} for the appropriate model parameters.
\end{prop}

To conclude, observe that the effective endowments of agents $k=0,...,n-1$ (who will trade prior to $t_n$) at $t_{n}$
\begin{equation*}
\begin{split}
\wt{\E}^{k}_{a_{n}} &\dfn -\log\left(\condexpvs{e^{-\gamma_k \wh{\We}^{k}_{a_n,b_{n}} - \wt{\E}^{k,-}_{b_{n}}}}{\F^{n,k}_{a_n}}\right) = -\log\left(\condexpvs{e^{-\gamma_k \wh{\We}^{k}_{a_n,1}}}{\F^{n,k}_{a_n}}\right),
\end{split}
\end{equation*}
take the form
\begin{equation*}
\begin{split}
\wt{\E}^{0}_{t_{n}} &= \ell_{n}(t_{n},X_{t_{n}},\ol{H}_{n})+\check{\lambda}^n_0 + \check{\E}^0_n(X_{t_{n}},\ol{H}_{n}),\\
\wt{\E}^{k}_{t_{n}} &= \ell_{{n},k}(t_{n},X_{t_{n}},\ol{H}^{-k}_{n},G_k) + \check{\lambda}^n_k+ \frac{1}{2}\left(S^n(t_n,X_{t_n},\ol{H}_n)\right)'(C_k-E_k)\left(S^n(t_n,X_{t_n},\ol{H}_n)\right)\\
&\qquad\ - \left(G_kC_k - E_k H_k\right)' S^n(t_n,X_{t_{n}},\ol{H}_{n}) + \check{\E}^{0}_n(X_{t_{n}},\ol{H}_{n}).
\end{split}
\end{equation*}
As before, $\check{\E}^{0}_{n}$ is linear-quadratic in $(X_{t_n},\ol{H}_n)$, and from Lemma \ref{L:sp_ass_verify} we know Assumption \ref{A:lq_sp} holds.

\subsection*{On Assumptions \ref{A:cont_time_verify} and \ref{A:lq_sp}}

To finish the proof of Theorem \ref{T:main_result}, we verify Assumption \ref{A:cont_time_verify} for $n=1,...,N$ and Assumption \ref{A:lq_sp} for $n=2,...,N$ and $k\leq n$. This lemma is proved in Appendix \ref{AS:last_period}.

\begin{lem}\label{L:sp_ass_verify}
Assumptions \ref{A:cont_time_verify} and \ref{A:lq_sp} holds for $n=1,...,N$ and $k\leq n-1$.
\end{lem}


\section{An Example with Two Signal Times}\label{S:One_jump}

Consider when there are two insiders who obtain private signals at $t_1 = 0$ and $t_2 \in (0,1)$.  According to Theorem \ref{T:main_result}, in equilibrium there is a jump at $t_2$ in the asset price (and  associated volatility, market price of risk), information flow, and optimal trading strategies. To visualize this, we perform a numerical analysis. Agent weights and risk aversions, the SDE parameters, and the supply are
\begin{equation}\label{E:parameter}
\begin{split}
\omega_j &= 1/3, \gamma_j = 3\quad j\in {0,1,2};\qquad \kappa = 1,\theta = 0,\sigma = 1,X_0 = 0;\qquad \Pi = 1.
\end{split}
\end{equation}
As for the second signal time and precisions, for the plots in Figure \ref{F:Trading_SP} we use
\begin{equation}\label{E:parameter_1}
C_1 = 1, C_2 = 2, D_1 = D_2 = 1, t_2 = 1/2,
\end{equation}
so that the second insider obtains a more precise signal.

\begin{figure}
\centering
\includegraphics[height=3cm,width=4.5cm]{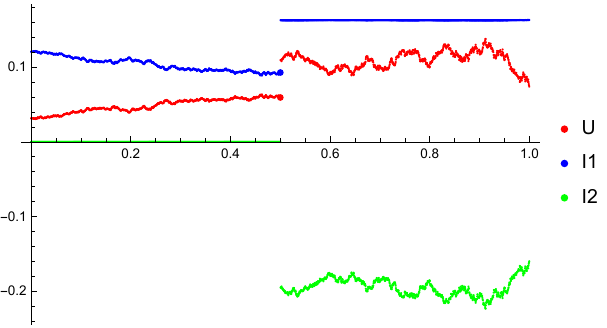}\ \includegraphics[height=3cm,width=4.5cm]{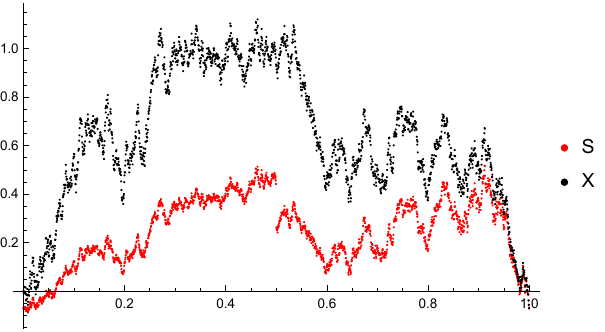}\ \includegraphics[height=3cm,width=4.5cm]{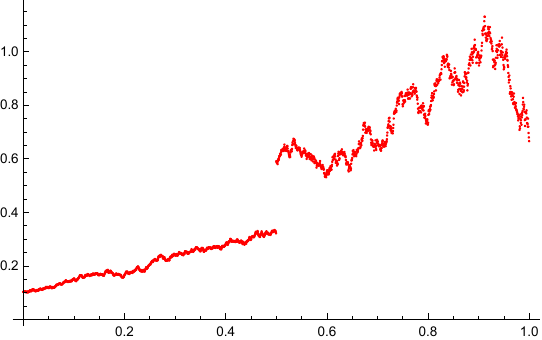}
\caption{From left to right, sample paths for the optimal trading strategies; the factor and price; and the market price of risk. Parameters are in \eqref{E:parameter} and \eqref{E:parameter_1}.}
\label{F:Trading_SP}
\end{figure}

\subsection*{Prices, strategies and risk premia} In Figure \ref{F:Trading_SP}, notice the jump in strategies for agents $0,1$ at $t_2$, but also how flat the strategy for agent $1$ is from $1/2$ to $1$.  This follows from the formulas for the strategies in Theorem \ref{T:main_result}. Indeed, in the univariate case when all agents have the same weight and risk aversion (hence common weighted risk tolerance $\alpha$), and when the noise trader precision $D_n=D$ is the same for all times, the coefficient in front of $S_t$ for $I_1$ takes the values
\begin{equation*}
\begin{split}
t\in [t_1,t_2]:\hspace{5pt} -\frac{\alpha C_1}{1+\alpha^2 D C_1},\qquad t\in (t_2,1]: \hspace{5pt} -\frac{\alpha C_1}{1+\alpha^2 D C_1}\left(1 -\frac{C_2-C_1}{C_1(1+\alpha^2 DC_2)}\right).
\end{split}
\end{equation*}
Thus, if $C_2 - C_1\approx C_1$, the position size will decrease dramatically, at least for small $\alpha$, as is the case using \eqref{E:parameter_1}.

\subsection*{Welfare} We next consider agent welfare, first remarking that a proper welfare analysis deserves a separate treatment. This is because the central question ``is it better to obtain a worse signal now or a better signal later?'' requires one to carefully model the relationship between the signal time, signal precision, and signal cost.

To motivate the issue, fix $t_2<1$ and note that by simply observing the public factor $X_{t_2}$, one obtains a signal about $X_1$ with precision $P(1-t_2) \approx (1-t_2)^{-1}$.  As such, if $t_2$ varies but the precision $C_2$ is fixed, then as $t_2 \to 1$ the signal $G_2$ is essentially useless. Similarly, motivated by the assumptions in \cite{back1992insider}, the noise trader precision $D_2$ should also be on the order $(1-t_2)^{-1}$.  Lastly, though un-modeled (because it does not affect equilibrium prices and trading strategies) there is presumably a cost to obtaining the signal, and if the second insider is receiving an ever more precise signal, this cost may become prohibitively large.  This cost will affect welfare, and it is not immediately clear how large it should be.

Taking account the above warnings, we compare ex-ante welfare from trading over $[t_2,1]$ across the agents. The idea is the second insider wants to enter at the time when she will outperform the first insider over the trading window $[t_2,1]$ to the greatest extent. This corresponds to when she is paid based upon  performance relative to her peers when she is active, and as such she ignores the time interval $[0,t_2]$\footnote{Both insiders outperform the uninformed agent over $[t_2,1]$ as they have finer information sets and larger acceptable trading strategy sets.}.

Following \cite{L-1985}, we define welfare over $[t_2,1]$ through the ex-ante certainty equivalent
\begin{equation*}
  \gamma_k \textrm{CE}^k_{0-} \dfn -\log \left( \expvs{ e^{- \gamma_k \wh{\We}^k_{t_2,1}}}\right);\qquad k= 0,1,2,
\end{equation*}
where the ``$0-$'' emphasizes the expectation is unconditional.

In Figure \ref{F:Welfare} we plot insider welfare above that of the uninformed agent (i.e. $\textrm{CE}^k_{0-}- \textrm{CE}^0_{0-}$ for $k=1,2$) as a function of the second signal time $t_2$. The parameters are given in \eqref{E:parameter}, and we also set $C_1 = 2, D_1 = 1$. The idea here is that, unless $C_2$ varies, it is worse than $C_1$. The left plot holds $C_2 = D_2=1$ constant; the middle plot holds $D_2=1$ constant while $C_2 = (1-t_2)^{-1}$; the right plot has $C_2 =D_2 = (1-t_2)^{-1}$.

While the second insider's entry time has only a small affect on the first insider's performance, the situation is much different for the second insider. Indeed, in the left plot of Figure \ref{F:Welfare}, both $C_2,D_2$ are fixed, and hence the first insider (with a better signal) outperforms.  In the middle plot, $D_2$ is fixed, but the second insider's signal is improving. Here, the second insider is motivated to wait because her increased precision improves welfare, and because the second market signal $H_2$ is not becoming precise due to non-vanishing noise. In fact, a rather perverse phenomena arises in that the insider might never enter, because she always thinks she can do better if she waits and ultimately runs out of time. Lastly, in the right plot, we allow the noise trader demand to become precise as well. Here, as the market signal precision $E_2$ is on the same order as $D_2$, the insider is motivated to enter prior to the end of the period, as her relative performance declines after a certain time.  Clearly, this simple example motivates a thorough welfare analysis, as absent other criteria one can obtain almost any conclusion one wants.

\begin{figure}
\centering
\includegraphics[height=3cm,width=4.5cm]{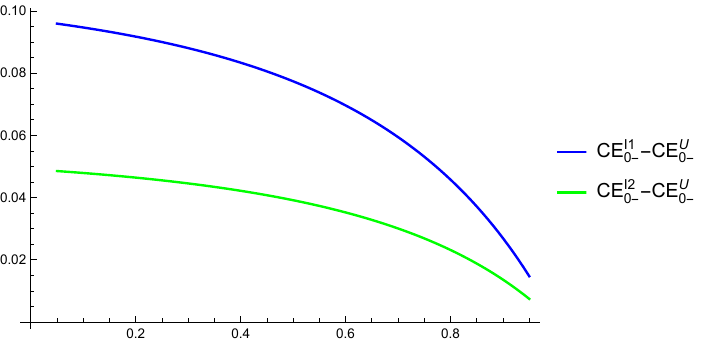}\ \includegraphics[height=3cm,width=4.5cm]{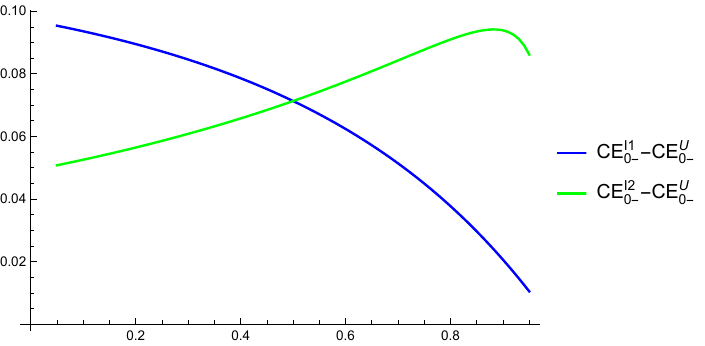}\ \includegraphics[height=3cm,width=4.5cm]{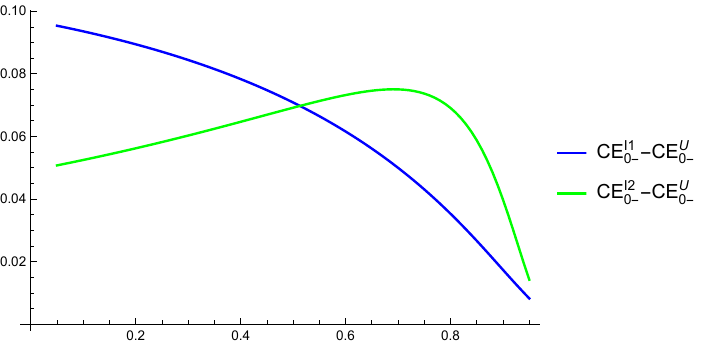}
\caption{From left to right, insider welfare above the uninformed over $[t_2,1]$ as a function of the insider's arrival time $t_2$. Parameters are in \eqref{E:parameter} along with $C_1=2,D_1=1$. The left plot has $C_2=D_2 = 1$; the middle plot has $C_2 = (1-t_2)^{-1}, D_2 = 1$; the right plot has $C_2 = D_2 = (1-t_2)^{-1}$.}
\label{F:Welfare}
\end{figure}


\section{The Large Number of Insiders Limit}\label{S:CJ}

We now consider when there are a large number of agents with private information entering the market throughout time.  Our goal is to prove that, under appropriate scaling, the limiting public filtration is $\filt^m = \filt^{B,J}$ where $J_\cdot = X_1 + Y_\cdot$ with $Y$ an independent Gaussian process.  This in turn may simplify computation of equilibrium quantities (already for $N=2$ signals, welfare computations are quite complex). Additionally, this would show the market filtration coincides with that considered in \cite{MR2213260}. However, while therein the filtration was exogenously assumed to be that of an insider who observes a private signal flow, here we show this type of filtration \emph{endogenously} arises as the \emph{public} filtration when there are a large number of insiders entering the market throughout time. Proofs of all results in this section are given in Appendix \ref{AS:CJ}.

Throughout, we make the following assumptions on the $N^{th}$ market, where the Brownian motions $B^Y$ and $B^Z$ below are independent of each other, as well as the Brownian motion $B$ driving $X$.  Additionally, for all the functions below, we require continuity at $0$ because the first insider's signal is always at time $0$. As there is no signal at $1$ we do not necessarily require continuity at $1$. In fact, as discussed at the end of this section, delicate issues arise at $t=1$.  Lastly, we define a typical dyadic rational in $[0,1)$ as
\begin{equation}\label{E:dyadic}
r^N_n \dfn \frac{n-1}{2^N};\qquad N=1,2,\dots;\quad n=1,\dots, 2^N.
\end{equation}

\begin{ass}\label{A:converge}\text{}
\begin{enumerate}[(1)]
\item (Signal times) There are $2^N$ insiders, with private signal times $t^N_n = \tau(r^N_n), n=1,...,2^N$. $\tau \in C[0,1]\cap C^1[0,1)$ is strictly increasing with $\tau(0) = 0$ and $\tau(1) = 1$.
\item (Risk aversion) The uninformed agent has risk aversion $\gamma_0$ independent of $N$.  The $n^{th}$ insider has risk aversion $\gamma^N_n = \gamma(t^N_n)$. $\gamma \in C[0,1)$ is strictly positive.
\item (Insider signal precision) The $n^{th}$ insider has precision matrix  $C^N_n = p(t^N_n) C_I$ where $C_I\in\sdpos$ is fixed. $p\in C[0,1)$ is a strictly positive  scalar function. To enforce this we assume the insider signal dispersions $Y^N_n$ take the form
    \begin{equation*}
    Y^N_n = p\left(t^N_n\right)^{-1/2} C_I^{-1/2}\left(B^{Y}_n - B^{Y}_{n-1}\right).
    \end{equation*}
\item (Noise trading) The $n^{th}$ incremental noise trader demand has precision $D^N_n = (1/(t^N_{n+1}-t^N_{n}))D_Z$ where $D_Z\in\sdpos$. To enforce this, we assume the noise trader demands $Z^N_n$ take the form
\begin{equation*}
    Z^{N}_n = D_Z^{-1/2}\int_{r^N_n}^{r^N_{n+1}} \sqrt{\dot{\tau}(u)}dB^Z_u,
\end{equation*}

\item (Agent weights) The weight of the uninformed agent $\omega_0$ is independent of $N$. Insiders $n=2,...,2^N$ have weights $\omega^N_{n} = 2^{-N}\omega(t^N_n)$, and the first insider has weight
\begin{equation*}
    \omega^N_1 = 1 - \omega_0 - \frac{1}{2^N}\sum_{n=2}^{2^N} \omega(t^N_n).
\end{equation*}
$\omega \in C[0,1)$ is strictly positive and such that $\int_0^1 \omega(\tau(u))du < 1-\omega_0$. Thus, for $N$ large enough
\begin{equation}\label{E:omega_n1_large}
    \omega^N_1 > 0;\qquad \lim_{N} \omega^N_1 = \omega_1 = 1 -\omega_0 - \int_0^1 \omega(\tau(u))du.
\end{equation}
\end{enumerate}
\end{ass}

\begin{rem} Let us discuss Assumption \ref{A:converge}, which is in force throughout this section. The function $\tau$ in $(1)$ allows for flexibility in modeling when insiders enter the market, as it may not be realistic to assume insiders arrive uniformly in time. We require additional smoothness on $\tau$ as well as $\dot{\tau}(0) > 0$ to rule out the first insider being ``separated'' from all the other insiders in the limit. For example, one can take $\tau$  to be concave, to reflect that more insiders enter the market sooner.

In $(3)$, it is not required, but realistically one should assume $p$ is increasing to $\infty$ as $t\to 1$. This reflects the idea that the signal becomes more precise the later it is received.

As for the weights in $(5)$,  we want to keep a fraction of the agents uninformed, even in the limit, as there should remain market participants who do not have direct access to private signals.  As for the first insider's weight, $t_0 = 0$ is when the flow of information starts, and (as will be shown below) \eqref{E:omega_n1_large} is needed to ensure that in the limit, there is an immediate impact on equilibrium quantities.  If, for example, we equally weighted all the insiders and equally spaced all signals, then as $N\to \infty$, the market signal precision matrix $E^N_1$ at time $0$ would go to $0$, leading to divergence of $H^N_1$.  However, by keeping a mass of insiders at $0$ we ensure a limiting signal at $0$.

Lastly, as shown in \eqref{E:ll_def} below, the functions $p,\omega,\gamma$ affect limiting quantities only through the risk tolerance weighted precision function $\lambda = p\omega/\gamma$, so  effectively there are only two free functions: $\lambda$ and $\tau$.
\end{rem}

\subsection{The Effective Signal Process}   In the $N^{th}$ market, assume $t \in [t^N_n,t^N_{n+1})$. The public information is $\F^{N,0}_t = \F^B_t \vee \sigma(\ol{H}^N_n)$ and using Lemma \ref{L:Lambda_ident} one can show
\begin{equation*}
\begin{split}
\frac{\condprobs{X_1\in dx}{\F^{N,0}_t}}{\condprobs{X_1\in dx}{\F^B_{t}}} &= \frac{e^{\Lambda\left(1,X_{1},\sum\limits_{m=1}^{n} E^N_m,\sum\limits_{m=1}^{n} E^N_m h^N_m \right)}}{e^{\Lambda\left(t,X_{t},\sum\limits_{m=1}^{n} E^N_m,\sum\limits_{m=1}^{n} E^N_m h^N_m \right)}}.
\end{split}
\end{equation*}
This motivates us to define the effective market precision and signal (compare with \eqref{E:N_matrices_def})
\begin{equation}\label{E:agg_Sig}
F^N_n \dfn \sum_{m=1}^n E^N_m;\qquad J^N_n \dfn \left(F^N_n\right)^{-1}\sum_{m=1}^n E^N_m H^N_m.
\end{equation}
As $\sigma(\ol{H}^N_n) = \sigma(H^N_1,...,H^N_n) = \sigma(J^N_1,...,J^N_n) = \sigma(\ol{J}^N_n)$, receiving the vector of signals $\ol{H}^N_{n}$ provides the same value, in the sense of estimating the law of $X_1$ in  as receiving the single signal $J^N_{n}$ and we can use $\ol{J}^N$ to generate the public filtration as well.  It is also advantageous to use $F^N,J^N$ because of Proposition \ref{P:agg_Sig_conv} below, which shows that $J^N$ converges to a signal process.  To state the proposition, for $t\in [0,1)$ (and setting $t^N_{2^N+1} = 1$) define the matrix and vector valued functions
\begin{equation}\label{E:FNt_JNt_def}
    \begin{split}
     F^N_t &\dfn F^{N}_{n}, \quad J^N_t \dfn J^{N}_{n}\qquad \textrm{for}\quad  t^N_n \leq t < t^N_{n+1}, n=1,...,2^N.
     \end{split}
\end{equation}
Next, define the risk-tolerance weighted precision function
\begin{equation}\label{E:ll_def}
\lambda(t) \dfn \frac{p(t)\omega(t)}{\gamma(t)};\qquad 0\leq t <  1.
\end{equation}
With this notation we have the convergence result

\begin{prop}\label{P:agg_Sig_conv} For $t \in [0,1)$ fixed
\begin{equation}\label{E:precision_signal_process}
\begin{split}
\lim_{N\to\infty} F^N_{t} &= F_t \dfn p(0)C_I + \left(\int_0^{\tau^{-1}(t)}\frac{\lambda(\tau(u))^2}{\dot{\tau}(u)}du\right)C_I D_{Z} C_I,\\
\lim_{N\to \infty} J^N_{t} &= J_t \dfn X_1 + F_{t}^{-1}\left(p(0)^{1/2} C_I^{1/2} B^Y_1 + \int_0^{\tau^{-1}(t)} \frac{\lambda(\tau(u))}{\sqrt{\dot{\tau}(u)}}C_I D_Z^{\frac{1}{2}}dB^{Z}_u\right).
\end{split}
\end{equation}
The second limit is in $L^2$.
\end{prop}

Given this, we expect that in the large insider limit, on $[0,1)$ the market filtration is the right-continuous augmentation $B$ and $J$'s natural filtration.

\subsection{Convergence Analysis over $[0,1)$}

We now take the limit $N\to\infty$ and show that on $[0,1)$ the limiting market filtration converges. To ease the notational burden, define
\begin{equation*}
    \filtg^{N,n} \dfn \filt^B \vee \sigma(\ol{H}^N_n) = \filt^B \vee \sigma(\ol{J}^N_n);\qquad n=1,...,2^N;\qquad \filtg_0 \dfn \filt^{B,J},
\end{equation*}
so that $\filtg^{N,n}$ is the market filtration in the $N^{th}$ market over $[t^N_n,t^N_{n+1})$. Now, recall that our filtrations are $\prob$-augmentations and that \cite[Proposition 3.3]{MR1775229} implies $\filtg^{N,n}$ is right-continuous.  However, $\filtg_0$, as the augmented natural filtration of $B,J$, is not necessarily right-continuous.  Therefore, we define $\filtg$ as the right continuous enlargement over $[0,1)$.

Using \eqref{E:Lambda_def} and \eqref{E:ell_def_0}, for $t\in [0,1)$ we have
\begin{equation*}
\begin{split}
\nabla_x \ell^N_n(t,X_t,\ol{h}^N_n) &=e^{-(1-t)\kappa'}P(1-t)\left(P(1-t)+F^N_n\right)^{-1}F^N_n\left(j^N_n - \mu(X_t,1-t)\right),
\end{split}
\end{equation*}
where $j^N_n = (F^N_n)^{-1}\sum_{m=1}^n E^N_m h^N_m$. Therefore, define
\begin{equation}\label{E:theta_Nn_def}
\begin{split}
\theta^{N,n}_t(j^N_n) &\dfn \sigma'e^{-(1-t)\kappa'}P(1-t)\left(P(1-t)+F^N_n\right)^{-1}F^N_n\left(j^N_n - \mu(X_t,1-t)\right),
\end{split}
\end{equation}
so that (abusing notation) $\theta^{N,n}_t(j^N_n) = \sigma'\nabla_x \ell^N_n(t,X_n,\ol{h}^N_n)$. With this notation we have
\begin{lem}\label{L:N_BM_lem}
The process
\begin{equation*}
B^{N,n}_\cdot \dfn B_\cdot - \int_0^\cdot \theta^{N,n}_u(J^N_n) du,
\end{equation*}
is a $\left(\prob,\filtg^{N,n}\right)$ Brownian motion on $[0,1]$.
\end{lem}

To state the next lemma, on $[0,1]$ define the filtration $\filtg^N$ by
\begin{equation*}
    \G^N_1 \dfn G^{N,2^N}_1, \quad \G^N_t \dfn \G^{N,n}_t \qquad \textrm{for} \quad t^N_n \leq t < t^N_{n+1},\  n=1,...,2^N,
\end{equation*}
and the process $\theta^N$ by
\begin{equation}\label{E:theta_N_def}
    \begin{split}
        \theta^N_0 \dfn \theta^{N,1}_0(J^N_1), \quad \theta^N_t \dfn \theta^{N,n}_t(J^N_n)\qquad \textrm{for}\quad  t^N_n < t \leq t^N_{n+1},\  n=1,...,2^N.
    \end{split}
\end{equation}
The next lemma shows that $\theta^N$ is ``information drift'' in the language of \cite{MR1632213,MR2213260}.

\begin{lem}\label{L:N_BM_lem2}
The process
\begin{equation*}
B^N_\cdot \dfn B_\cdot - \int_0^\cdot \theta^N_u du,
\end{equation*}
is a $\left(\prob,\filtg^N\right)$ Brownian motion on $[0,1]$.
\end{lem}

We now prove convergence of $\filtg^{N}$ to $\filtg_0$ over $[0,1)$, and also obtain the $(\prob,\filtg)$ semi-martingale decomposition for $B$. As a first step,
and analogously to \eqref{E:precision_signal_process}, define
\begin{equation}\label{E:theta_def}
\theta_t \dfn \sigma'e^{-(1-t)\kappa'}P(1-t)\left(P(1-t)+F_t\right)^{-1}F_t\left(J_t - \mu(X_t,1-t)\right);\qquad t \in [0,1).
\end{equation}
Using Proposition \ref{P:agg_Sig_conv}, we obtain

\begin{lem}\label{L:theta_converge_p}
For each fixed $t \in [0,1)$  $\theta^N_t$ converges to $\theta_t$ in probability.
\end{lem}

Lemma \ref{L:theta_converge_p} suggests that $B_\cdot - \int_0^\cdot \theta_u du$ might be a $(\prob,\filtg)$ Brownian motion.  To prove this, we need to prove $\theta^N \to \theta$ in a stronger sense than in probability for each fixed $t$.  To do so, we claim

\begin{prop}\label{P:l2bdd}
For each $t\in [0,1)$ fixed $\sup_{N} \expvs{\int_0^t \left|\theta^{N}_u\right|^2 du} < \infty$.
\end{prop}

Next, we recall the definition of $L^p$ convergence for sigma-algebras and filtrations from \cite{MR1113218, MR1739298, coquet2000some, MR1837294} as well as the more recent \cite{MR3358108,neufcourtthesis}.

\begin{defn}\label{D:filt_converge}
A sequence of $\sigma$-algebras $\cbra{\G^N}$ converges to a $\sigma$-algebra $\G$ in $L^p$ if either one of the equivalent conditions holds.
\begin{enumerate}[(i)]
\item For all $B\in\G$ we have $\lim_{N\to\infty} \condexpvs{1_B}{\G^N} = 1_B$ in $L^p$.
\item For all $Y\in L^p(\G)$ we have $\lim_{N\to\infty} \condexpvs{Y}{\G^N} = Y$ in $L^p$.
\end{enumerate}
A sequence of filtrations $\filtg^N$ converges to a filtration $\filtg$ in $L^p$ over $[0,1)$ if $\G^N_t \to \G_t$ in $L^p$ for all $t \in [0,1)$.
\end{defn}

Our first filtration convergence result is

\begin{prop}\label{P:filt_converge} $\filtg^N$ converges to $\filtg_0$ in $L^p$ over $[0,1)$ for any $p\geq 1$.
\end{prop}

With this result we obtain the main result of the section

\begin{thm}\label{T:limit_bm}
$B_\cdot - \int_0^\cdot \theta_u du$ is both a $(\prob,\filtg_0)$ and $(\prob,\filtg)$ Brownian motion over $[0,1)$.
\end{thm}

\subsection{What Happens at $t=1$?}  Come back to \eqref{E:FNt_JNt_def}.  As $v \to \int_0^v \lambda(\tau(u))^2/\dot{\tau}(u)du$ is increasing, it is clear that
\begin{equation*}
    F^N_1 \dfn \lim_{t\to 1} F^N_t = p(0)C_I + \left(\int_0^1 \frac{\lambda(\tau(u))^2}{\dot{\tau}(u)}du\right)C_I D_Z C_I,
\end{equation*}
is well defined even if
\begin{equation}\label{E:int_cond}
    \int_0^1 \frac{\lambda(\tau(u))^2}{\dot{\tau}(u)}du = \infty.
\end{equation}
Similarly, by the strong law of large numbers for Brownian motion, $\lim_{t\to 1} J^N_t$ is well defined. If \eqref{E:int_cond} holds then $J_1 = \lim_{t\to 1} J_t = X_1$, else
\begin{equation*}
    J_1 = \lim_{t\to 1} J_t = X_1 + F_{1}^{-1}\left(p(0)^{1/2} C_I^{1/2} B^Y_1 + \int_0^{1} \frac{\lambda(\tau(u))}{\sqrt{\dot{\tau}(u)}}C_I D_Z^{\frac{1}{2}}dB^{Z}_u\right).
\end{equation*}
Therefore, we can extend $J$ to $1$ and hence define $\filtg_0, \filtg$ for $t=1$ as well\footnote{$J_1$ does not provide any useful information as the market sees $X_1$ anyway, and the only other possible term in $J_1$ is a noise term independent of the asset terminal payoff.}. Additionally, the process  $\wt{B}_\cdot \dfn B_\cdot - \int_0^{\cdot}$ is well defined provided $\int_0^1 |\theta_t|dt < \infty$ a.s., and this will hold provided $\expvs{\int_0^1 |\theta_t|^2 dt}< \infty$.  As
\begin{equation*}
    J_1 - \mu(X_t,1-t) \sim N\left(0, \Sigma(1-t) + F_t^{-1}\right),
\end{equation*}
we see from \eqref{E:theta_def} that
\begin{equation*}
    \expvs{\int_0^1 |\theta_t|^2 dt}< \infty \quad \Longleftrightarrow \quad \int_0^1 \mathrm{Tr}\left(\left(\Sigma(1-t) + F_t^{-1}\right)^{-1}\right)dt < \infty.
\end{equation*}
To see when this holds, define
\begin{equation}\label{E:l2bddq}
Q \dfn \int_{0}^1 \frac{ \dot{\tau}(t)\int_0^t \frac{\lambda(\tau(u))^2}{\dot{\tau}(u)} du}{1 + (1-\tau(t))\int_0^t \frac{\lambda(\tau(u))^2}{\dot{\tau}(u)} du}dt.
\end{equation}
Then, we have the following.

\begin{prop}\label{P:teq1}
$\expvs{\int_0^1 |\theta_t|^2 dt}< \infty $ if and only if $Q<\infty$. If $Q<\infty$ then $\wt{B}_1$ is well defined, and $\wt{B}$ is both a $(\prob,\filtg_0)$ and $(\prob,\filtg)$ Brownian motion on $[0,1]$.
\end{prop}

\begin{rem}
The above suggests that by assuming $Q<\infty$ one could directly prove the martingale property for $\wt{B}$ on $[0,1]$ using the same arguments that we did over $[0,t]$ for $t<1$.  However, while $Q=\infty$ implies $\sup_N \expvs{\int_0^1 |\theta^N_t|^2 dt} = \infty$, it is not so clear $Q<\infty$ implies $\sup_N \expvs{\int_0^1 |\theta^N_t|^2 dt} < \infty$.  One must appeal to an integral convergence theorem to deduce this result and, aside from making draconian assumptions such as $\lambda, \dot{\tau}$ being bounded on $[0,1]$ (which would rule out $J_1 = X_1$), it is not obvious how to proceed.  However, this is not necessary. We can obtain the limiting Brownian motion and market filtration on $[0,1)$ under mild conditions, and then extend to $t=1$ at the post-limit level.
\end{rem}

We conclude by considering the equally spaced case when $\tau(u) = u$. If we want $J_1 = X_1$ and $\wt{B}$ to be a Brownian motion on $[0,1]$ then  we need \eqref{E:int_cond} and $Q<\infty$ for $Q$ from \eqref{E:l2bddq}. Direct calculations show that if $\lambda(t) = (1-t)^{-\lambda}$ then this holds if and only if $1/2 \leq \lambda < 1$.

\appendix

\section{On the Admissible Trading Strategies}\label{AS:strat}

Recall the admissible classes in \eqref{E:lp_admit}. Our goal is to invoke \cite[Propositions 3.4, 3.5]{detemple2022dynamic}, which shows the classes $\A^{N,0}_{t_N,1}, \A^{N,k}_{t_N,1} $ coincide, up to indistinguishability of the related wealth processes, with the ``typical'' admissible classes
\begin{equation*}
\begin{split}
    \wt{A}^{N,0}_{t_N,1} &\dfn \bigg\{ \pi\in\mcp(\filt^{N,0}) \such \pi \textrm{ is } S^{N,\ol{H}_N} \textrm{ integrable over }[t_N,1], \condexpv{\qprob^{N,0}}{}{\int_{t_N}^1 \pi_t'dS^{N,\ol{H}_N}_t}{\F^{N,0}_{t_N}} \leq 0\bigg\},\\
    \wt{A}^{N,k}_{t_N,1} &\dfn \bigg\{ \pi\in\mcp(\filt^{N,k}) \such \pi \textrm{ is } S^{N,\ol{H}_N} \textrm{ integrable over }[t_N,1], \condexpv{\qprob^{N,k}}{}{\int_{t_N}^1 \pi_t'dS^{N,\ol{H}_N}_t}{\F^{N,k}_{t_N}} \leq 0\bigg\}.
\end{split}
\end{equation*}
However, to even define these latter classes of processes, one must show $S^{N,\ol{H}_N}$ is a $(\filt^{N,0},\prob)$ semi-martingale; one must define the martingale measures $\qprob^{N,0},\qprob^{N,k}$ on $\F^{N,0}_{1},\F^{N,k}_{1}$ respectively; and one must verify the conditional expectations
\begin{equation*}
\condexpv{\qprob^{N,0}}{}{\int_{t_N}^1 \pi_t'dS^{N,\ol{H}_N}_t}{\F^{N,0}_{t_N}},\quad \condexpv{\qprob^{N,k}}{}{\int_{t_N}^1 \pi_t'dS^{N,\ol{H}_N}_t}{\F^{N,k}_{t_N}},
\end{equation*}
are well defined.  As this is rigorously done in \cite[Section 3, Appendices A,B]{detemple2022dynamic}, for the sake of brevity we outline the argument below, referring the reader to these sections for the proofs of the statements made.

For the martingale measures, Lemma \ref{L:Lambda_ident} verifies the Jacod  condition $\condprobs{\ol{H}_N\in d\ol{h}_N}{\F^B_t} \sim \prob\bra{\ol{H}_N\in d\ol{h}_N}$ almost surely for $t\in [0,1]$, and easy calculations (see \cite{MR1775229, MR3758346}) show $1/p^{N,0}(\cdot,X_{\cdot},\ol{H}_N)$ is a $(\prob,\filt^{N,0})$ martingale with constant expectation $1$. Thus,  we may define the $\filt^B$ to $\filt^{N,0}$ martingale preserving measure (c.f. \cite{MR1418248, MR1632213}) by
\begin{equation*}
\tprob^{N}\bra{A} \dfn \expvs{1_A \frac{1}{p^{N,0}(1,X_{1},\ol{H}_N)}};\qquad A\in \F.
\end{equation*}
Similarly, for the $k^{th}$ insider, Lemma \ref{L:Lambda_ident} verifies $\condprobs{\ol{G}_k\in d\ol{g}_k}{\F^{N,0}_t} \sim \prob\bra{\ol{G}_k\in d\ol{g}_k}$, $a.s.\ 0\leq t\leq 1$, and we define the $\filt^{N,0}$ to $\filt^{N,k}$ martingale preserving measure
\begin{equation*}
\tprob^{N,k}\bra{A} \dfn \expvs{1_A \left(\frac{p^{N,k}(1,X_1,\ol{H}^{-k}_N,G_k)}{p^{N,0}(1,X_1,\ol{H}_N)}\right)^{-1}};\qquad A\in \F.
\end{equation*}
With this notation, and recalling $Z^N$ from \eqref{E:Nmkt_density}, we define a candidate $(S^{N,\ol{H}_N},\filt^{N,0})$ martingale measure $\qprob^{N,0}$ by
\begin{equation}\label{E:U_mm}
    \frac{d\qprob^{N,0}}{d\prob}  = \frac{d\tprob^N}{d\prob} \times Z^N(1,X_1,\ol{H}_N) = \frac{Z^N(1,X_1,\ol{H}_N)}{p^{N,0}(1,X_1,\ol{H}_N)},
\end{equation}
and for the $k^{th}$ insider, we define a candidate $(S^{N,\ol{H}_N},\filt^{N,k})$ martingale measure $\qprob^{N,k}$ by
\begin{equation}\label{E:I_mm}
    \frac{d\qprob^{N,k}}{d\prob}  = \frac{d\qprob^{N,0}}{d\prob} \times \frac{p^{N,0}(1,X_1,\ol{H}_N)}{p^{N,k}(1,X_1,\ol{H}_N^{-k},G_k)} = \frac{Z^N(1,X_1,\ol{H}_N)}{p^{N,k}(1,X_1,\ol{H}_N^{-k},G_k)}.
\end{equation}

These are candidate martingale measures, because as shown in \cite[Lemma A.8]{detemple2022dynamic}, for all $0\leq s < t\leq 1$, and all $\filt^B_t \otimes \B(\reals^{d\times N})$ measurable random variables $\chi^{\ol{H}_N}_t$ such that $\chi^{\ol{h}_N}_t \in L^1(\qprob^{N,\ol{h}_N})$ for Lebesgue almost all $\ol{h}_N$ (respectively generic $\filt^B_t \otimes \B(\reals^{d\times (1+N)})$ measurable random variables $\chi^{G_k,\ol{H}_N}_t$ such that $\chi^{g_k,\ol{h}_N}_t \in L^1(\qprob^{N,\ol{h}_N})$ for Lebesgue almost all $(g_k,\ol{h}_N)$ ) we have
\begin{equation*}
\begin{split}
    \condexpv{\qprob^{N,0}}{}{\chi^{\ol{H}_N}_t}{\F^{N,0}_s} &= \left(\condexpv{\qprob^{N,\ol{h}_N}}{}{\chi^{\ol{h}_N}}{\F^B_s}\right)\big|_{\ol{h}_N = \ol{H}_N},\\
    \condexpv{\qprob^{N,k}}{}{\chi^{G_k,\ol{H}_N}_t}{\F^{N,k}_s} &= \left(\condexpv{\qprob^{N,\ol{h}_N}}{}{\chi^{g_k,\ol{h}_N}}{\F^B_s}\right)\big|_{(g_k,\ol{h}_N) = (G_k,\ol{H}_N)}.
\end{split}
\end{equation*}
Note that part of this statement is that the conditional expectations on the left side above are well defined in the generalized sense of \cite[Section 2.7]{MR3467826}.  Therefore over $[t_N,1]$, $S^{N,\ol{H}_N}$ satisfies the martingale property for both $(\qprob^{N,0},\filt^{N,0})$ and $(\qprob^{N,k},\filt^{N,k})$.   Lastly, as shown in \cite[Propositions 3.4, 3.5]{detemple2022dynamic} we have the intuitive measurability result
\begin{equation}\label{E:intuit_meas}
\begin{split}
\pi\in\A^{N,0}_{t_N,1} &\quad \Rightarrow\quad \int_{t_N}^{1} (\pi^{\ol{H}_N}_u)'dS^{N,\ol{H}_N}_u \quad \textrm{  is  } \F^B_1 \otimes\B(\reals^{d\times N}) \textrm{  measurable},\\
\pi\in\A^{N,k}_{t_N,1} &\quad \Rightarrow\quad \int_{t_N}^{1} (\pi^{G_k,\ol{H}_N}_u)'dS^{N,\ol{H}_N}_u \quad \textrm{  is  } \F^B_1 \otimes\B(\reals^{d\times(1+N)}) \textrm{  measurable},\\
\end{split}
\end{equation}
and intuitive relationship between wealth processes that
\begin{equation}\label{E:intuit_relation}
\begin{split}
\pi\in\A^{N,0}_{t_N,1} &\quad \Rightarrow\quad \int_{t_N}^{\cdot} (\pi^{\ol{H}_N}_u)'dS^{N,\ol{H}_N}_u = \left(\int_{t_N}^\cdot (\pi^{\ol{h}_N}_u)'dS^{N,\ol{h}_N}_u\right)\big|_{\ol{h}_N = \ol{H}_N},\\
\pi\in\A^{N,k}_{t_N,1} &\quad \Rightarrow\quad \int_{t_N}^{\cdot} (\pi^{G_k,\ol{H}_N}_u)'dS^{N,\ol{H}_N}_u = \left(\int_{t_N}^\cdot (\pi^{g_k,\ol{h}_N}_u)'dS^{N,\ol{h}_N}_u\right)\big|_{\ol{h}_N = \ol{H}_N}.
\end{split}
\end{equation}

\section{Proofs from Section \ref{S:ExpCaraEq}}\label{AS:ExpCaraEq}


\begin{proof}[Proof of Proposition \ref{P:abstract_cont_time}]
Assume for now the price process $S$ from \eqref{E:cont_price_def} is an Ito process with quadratic covariation $d\langle S, S \rangle_t = \Sigma^S_t dt$ where $\Sigma^S$ is deterministic.  Given this, for any symmetric matrix $E$ the random variable $(1/2) S_b ' E S_b$ can be replicated using the strategy $\pi_t = E S_t$ and initial capital $(1/2)S_a ' E S_a + (1/2)\int_a^b \mathrm{Tr}(E\Sigma^S_u)du$. Indeed, by integration by parts
\begin{equation*}
\begin{split}
\frac{1}{2} S_b'E S_b &= \frac{1}{2}\sum_{i,j} E^{ij}\left(S^i_a S^j_a + \int_a^b S^i_u dS^j_u + \int_a^b S^j_u dS^i_u + \int_a^b (\Sigma^S_u)^{ij} du\right),\\
&= \frac{1}{2}\left(S_a'ES_a + \int_a^b \mathrm{Tr}\left(E\Sigma^S_u\right)du\right) + \int_a^b (ES_u)'dS_u.
\end{split}
\end{equation*}
The result follows as $\Sigma^S$ is not random. Now, consider the optimal investment problem for agent $j$
\begin{equation*}
\inf_{\pi\in \A_{a,b}} \condexpvs{e^{-\gamma_j \We^{\pi}_{a,b} - 1_{j>0}\left(\frac{1}{2}S_b'M_jS_b - S_b'V_j\right) - \frac{1}{2}X_b'M X_b + X_b'V - \lambda_j}}{\F^B_a}.
\end{equation*}
Using the above, make the translation $\pi_t = (1/\gamma_j)\left(1_{j>0}\left(V_j - M_j S_t\right) + \theta_t\right)$. This leaves
\begin{equation*}
e^{1_{j>0}\left(V_j'S_a - \frac{1}{2}S_a'M_j S_a - \frac{1}{2}\int_a^b \mathrm{Tr}\left(M_j \Sigma^S_u\right)du\right)-\lambda_j}\times \inf_{\theta\in \A_{a,b}} \condexpvs{e^{-\We^{\theta}_{a,b} - \frac{1}{2}X_b'M X_b + X_b'V}}{\F^B_a},
\end{equation*}
along with the adjusted outstanding supply $\Psi_t = \Psi + \sum_{k=1}^J \alpha_k\left(M_k S_t - V_k\right)$ for $t\in [a,b]$. In the new optimization problem all agents have the same information and random endowment, and hence they all use the same strategy. In other words, for some common $\theta$ we have $\wh{\theta}^j_t =  \wh{\theta}_t$ for $j = 0, ..., J$. By market clearing $\Psi_t = \theta_t \sum_{j=0}^J \alpha_k = \theta_t / \ol{\gamma}$ and hence
\begin{equation*}
\begin{split}
\wh{\pi}^j_t &= \frac{1_{j>0}}{\gamma_j}\left(V_j - M_j S_t\right) + \frac{\ol{\gamma}}{\gamma_j} \Psi_t,\\
&= \frac{1_{j>0}}{\gamma_j}\left(V_j - M_j S_t\right) + \frac{\ol{\gamma}}{\gamma_j}\left(\Psi + \sum_{k=1}^J \alpha_k\left(M_k S_t - V_k\right)\right).
\end{split}
\end{equation*}
We now construct $Z_b$ which ensures $\wh{\theta}^j$ is indeed optimal for agent $j$. To do this, note that
\begin{equation*}
\begin{split}
\We^{\wh{\theta}^j}_{a,b} &= \ol{\gamma}\left(\Psi - \sum_{k=1}^J \alpha_k V_k\right)'(S_b-S_a) + \frac{1}{2}S_b'\left(\ol{\gamma}\sum_{k=1}^J \alpha_k M_k\right)S_b\\
&\qquad  - \frac{1}{2}S_a'\left(\ol{\gamma}\sum_{k=1}^J \alpha_k M_k\right)S_a - \frac{1}{2}\int_a^b \mathrm{Tr}\left(\left(\ol{\gamma}\sum_{k=1}^J \alpha_k M_k\right)\Sigma^S_u\right)du.
\end{split}
\end{equation*}
Using $S_b = \pvct + \pmtx X_b$ and the notation in \eqref{E:cont_olmv}
\begin{equation*}
\begin{split}
&\ol{\gamma}\left(\Psi - \sum_{k=1}^J \alpha_k V_k\right)'S_b + \frac{1}{2}S_b'\left(\ol{\gamma}\sum_{k=1}^J \alpha_k M_k\right)S_b\\
&\qquad = \frac{1}{2}X_b'\ol{M} X_b - X_b'\ol{V} + \frac{1}{2}\ol{\gamma}\pvct'\left(\sum_{k=1}^J \alpha_k M_k\right)\pvct - \ol{\gamma}\pvct'\left(\sum_{k=1}^J \alpha_k V_k - \Psi\right).
\end{split}
\end{equation*}
As such, the first order optimality condition $e^{-\We^{\wh{\theta}^j}_{a,b}-(1/2)X_b'MX_b + X_b'V} = \lambda^j(a) Z_b$, for some $\F^B_a$ measurable constant $\lambda^j(a)$ is satisfied provided
\begin{equation*}
\begin{split}
Z_b &= \frac{e^{-\frac{1}{2}X_b '\left(M+\ol{M}\right)X_b + X_b'\left(V + \ol{V}\right)}}{\expvs{e^{-\frac{1}{2}X_b '\left(M+\ol{M}\right)X_b + X_b'\left(V+\ol{V}\right)}}}.
\end{split}
\end{equation*}
Under Assumption \ref{A:cont_time_verify} $Z_b$ and $S$ from \eqref{E:cont_price_def} are clearly well-defined, and direct calculations using the normality of $X_b$ given $\F^B_t$ (see \eqref{E:Lambda_def}) show for $t\in [a,b]$
\begin{equation}\label{E:cont_px}
S_t = \pvct + \pmtx \frac{\condexpvs{S_b Z_b}{\F^B_t}}{\condexpvs{Z_b}{\F^B_t}} = \pvct + \pmtx\left(P(b-t) + M  + \ol{M}\right)^{-1}\left(P(b-t)\mu(X_t,b-t) + V + \ol{V}\right).
\end{equation}
Assumption \ref{A:cont_time_verify} also implies the $(S,\filt^B)$ market is complete over $[a,b]$, and as $\mu$ is linear in $X$ and $X$ has deterministic volatility, Ito's formula implies $d\langle S, S \rangle_t = \Sigma^S_t$ for non-random $\Sigma^S$. The remainder of the proof identifies the certainty equivalents. After the translation we have
\begin{equation*}
\begin{split}
e^{-\gamma_j \check{\E}^j_a} &= e^{-\lambda_j + 1_{j>0}\left(V_j'S_a - \frac{1}{2}S_a'M_j S_a - \frac{1}{2}\int_a^b \mathrm{Tr}\left(M_j \Sigma^S_u\right)du\right)}\times \condexpvs{e^{-\We^{\wh{\theta}^j}_{a,b} - \frac{1}{2}X_b'M X_b + X_b'V}}{\F^B_a},
\end{split}
\end{equation*}
so that using \eqref{E:Lambda_def}
\begin{equation*}
\begin{split}
\gamma_j \check{\E}^j_a &= \lambda_j + 1_{j>0}\left(\frac{1}{2}S_a'M_j S_a - S_a'V_j + \frac{1}{2}\int_a^b \mathrm{Tr}\left(M_j \Sigma^S_u\right)du\right) + \ol{\gamma}\left(\sum_{k=1}^J \alpha_k V_k-\Psi\right)'S_a\\
&\quad  - \frac{1}{2}S_a'\left(\ol{\gamma}\sum_{k=1}^J \alpha_k M_k\right)S_a - \frac{1}{2}\int_a^b \mathrm{Tr}\left(\left(\ol{\gamma}\sum_{k=1}^J \alpha_k M_k\right)\Sigma^S_u\right)du\\
&\quad  + \frac{1}{2}\ol{\gamma}\pvct'\left(\sum_{k=1}^J \alpha_k M_k\right)\pvct - \ol{\gamma}\pvct'\left(\sum_{k=1}^J \alpha_k V_k - \Psi\right)\\
&\quad  - \log\left(\condexpvs{e^{-\frac{1}{2}X_b '(M+\ol{M})X_b + X_b'\left(V+\ol{V}\right)}}{\F^B_a}\right).
\end{split}
\end{equation*}
The normality of $X_b$ given $\F^B_a$ implies (here and going forward, we write $P$ for $P(b-a)$ and $\mu$ for $\mu(X_a,b-a)$ to ease notation)
\begin{equation*}
\begin{split}
&-\log\left(\condexpvs{e^{-\frac{1}{2}X_b '(M+\ol{M})X_b + X_b'\left(V+\ol{V}\right)}}{\F^B_a}\right)\\
&\qquad = \frac{1}{2}\log\left(|1_d + P^{-1}(M+\ol{M})|\right) + \frac{1}{2}\mu'P\mu - \frac{1}{2}\left(P\mu+V+\ol{V}\right)'\left(P+M+\ol{M}\right)^{-1}\left(P\mu + V + \ol{V}\right).
\end{split}
\end{equation*}
Define the constant
\begin{equation}\label{E:cont_lambda_def}
\begin{split}
\check{\lambda}_j &\dfn 1_{j>0}\left(\int_a^b \mathrm{Tr}\left(M_j \Sigma^S_u\right)du\right) + \frac{1}{2}\log\left(|1_d + P^{-1}(M+\ol{M})|\right)\\
&\qquad  - \frac{1}{2}\int_a^b \mathrm{Tr}\left(\left(\ol{\gamma}\sum_{k=1}^J \alpha_k M_k\right)\Sigma^S_u\right)du.
\end{split}
\end{equation}
This implies the certainty equivalent is
\begin{equation*}
\begin{split}
\gamma_j \check{\E}^j_a &= \lambda_j + \check{\lambda}_j +  1_{j>0}\left(\frac{1}{2}S_a'M_j S_a - S_a'V_j\right) +  \ol{\gamma}\left(S_a-\pvct\right)'\left( \sum_{k=1}^J \alpha_k V_k - \Psi\right) - \frac{1}{2}S_a'\left(\ol{\gamma}\sum_{k=1}^J \alpha_k M_k\right)S_a\\
&\qquad + \frac{1}{2}\pvct'\left(\ol{\gamma}\sum_{k=1}^J \alpha_k M_k\right)\pvct  + \frac{1}{2}\mu'P\mu - \frac{1}{2}\left(P\mu+V+\ol{V}\right)'\left(P+M+\ol{M}\right)^{-1}\left(P\mu + V + \ol{V}\right).
\end{split}
\end{equation*}
Let us write this as
\begin{equation*}
\gamma_j \check{\E}^j_a = \lambda_j + \check{\lambda}_j +  1_{j>0}\left(\frac{1}{2}S_a'M_j S_a - S_a'V_j\right) + \mathbf{Q}.
\end{equation*}
From \eqref{E:cont_px} we know $S_a = \pvct + \pmtx(P+M+\ol{M})^{-1}(P\mu+V+\ol{V})$, and by construction $\ol{\gamma}\pmtx'(\sum_{k=1}^J \alpha_k V_k - \Psi) = \ol{V} + \ol{\gamma}\pmtx'(\sum_{k=1}^J \alpha_k M_k)\pvct$. With $\mathbf{M} = P + M + \ol{M}$ this yields
\begin{equation*}
\begin{split}
\mathbf{Q} &= \left(P\mu+V+\ol{V}\right)'\mathbf{M}^{-1}\left(\ol{V} + \pmtx'\left(\ol{\gamma}\sum_{k=1}^J \alpha_k M_k\right)\pvct\right)\\
&\quad  - \frac{1}{2}\left(\pvct + \pmtx\mathbf{M}^{-1}\left(P\mu+V+\ol{V}\right)\right)'\left(\ol{\gamma}\sum_{k=1}^J \alpha_k M_k\right)\left(\pvct + \pmtx\mathbf{M}^{-1}\left(P\mu+V+\ol{V}\right)\right)+ \frac{1}{2}\mu'P\mu\\
&\quad + \frac{1}{2}\pvct'\left(\ol{\gamma}\sum_{k=1}^J \alpha_k M_k\right)\pvct  - \frac{1}{2}\left(P\mu+V+\ol{V}\right)'\mathbf{M}^{-1}\left(P\mu + V + \ol{V}\right),\\\
&= \frac{1}{2}\mu'P\mu - \frac{1}{2}\left(P\mu + V + \ol{V}\right)'\mathbf{M}^{-1}\ol{M}\mathbf{M}^{-1}\left(P\mu + V + \ol{V}\right)\\
&\quad - \frac{1}{2}\left(P\mu+V+\ol{V}\right)'\mathbf{M}^{-1}\left(P\mu + V + \ol{V}\right) + \left(P\mu+V+\ol{V}\right)'\mathbf{M}^{-1}\ol{V},\\
&= \frac{1}{2}\mu'P\mu - \frac{1}{2}\left(P\mu + V + \ol{V}\right)'\mathbf{M}^{-1}\left(P + M + 2\ol{M}\right)\mathbf{M}^{-1}\left(P\mu + V + \ol{V}\right)\\
&\quad + \left(P\mu+V+\ol{V}\right)'\mathbf{M}^{-1}\ol{V}
\end{split}
\end{equation*}
The above is evidently linear-quadratic in $(y=\mu,v,\ol{v})$ with coefficients
\begin{equation}\label{E:cont_lq_coeff}
\begin{split}
\frac{1}{2}y' \mathbf{E} y&: \quad\mathbf{E} = P\mathbf{M}^{-1}\left(M + (M+\ol{M})P^{-1}(M+\ol{M})\right)\mathbf{M}^{-1}P,\\
\frac{1}{2}v'\mathbf{E} v&: \quad\mathbf{E} = \mathbf{M}^{-1}(P+M)\mathbf{M}^{-1},\\
\frac{1}{2}\ol{v}'\mathbf{E} \ol{v}&: \quad\mathbf{E} = -\mathbf{M}^{-1}(P+M+2\ol{M})\mathbf{M}^{-1},\\
y' \mathbf{E} v&: \quad\mathbf{E} = -P\mathbf{M}^{-1}\ol{M}\mathbf{M}^{-1},\\
y'\mathbf{E} \ol{v}&: \quad\mathbf{E} = -P\mathbf{M}^{-1}(P+M+2\ol{M})\mathbf{M}^{-1},\\
v'\mathbf{E}\ol{v}&: \quad\mathbf{E} = -\mathbf{M}^{-1}\ol{M}\mathbf{M}^{-1},
\end{split}
\end{equation}
where, bringing back the function argument, $\mathbf{M} = P(b-a) + M + \ol{M}$.

\end{proof}


\begin{proof}[Proof of Proposition \ref{P:abstract_sp}]

Consider the optimal investment problem for agent $j$
\begin{equation*}
\inf_{\pi\in\reals^d} \expvs{e^{-\gamma_j \pi'(S-p) - 1_{j>0}\left(\frac{1}{2}S'M_jS - S'V_j\right) - \frac{1}{2}H'M H + H'V - \lambda_j}}.
\end{equation*}
Make the translation $\pi = (1/\gamma_j)(1_{j>0}V_j + \theta)$ to obtain
\begin{equation*}
e^{-\lambda_j + 1_{j>0}V_j'p} \times \inf_{\theta\in \reals^d} \expvs{e^{-\theta'(S-p) - 1_{j>0}\frac{1}{2}S'M_jS - \frac{1}{2}H'MH + H'V}},
\end{equation*}
along with the adjusted outstanding supply $\Psi - \sum_{k=1}^J \alpha_k V_k$. Plug in $Y = \pvct + \pmtx H$ to obtain
\begin{equation*}
\begin{split}
e^{-\lambda_j - 1_{j>0}\left(\frac{1}{2}\varphi'M_j\varphi - p'V_j\right)} \times \inf_{\theta\in \reals^d}e^{\theta'(p-\pvct)} \expvs{e^{-\frac{1}{2} H'\left(M+1_{j>0}\pmtx'M_j\pmtx\right)H + H'\left(V-1_{j>0}\pmtx'M_j\pvct - \pmtx'\theta\right)}}.
\end{split}
\end{equation*}
Using $H\sim N(0,\Sigma^H)$ this becomes
\begin{equation}\label{E:normal_pre_opt}
\begin{split}
&e^{-\lambda_j + 1_{j>0}\left(\frac{1}{2}\varphi'M_j\varphi - p'V_j\right)-\frac{1}{2}\log\left(\left|1_d + \Sigma^H(1_{j>0}\pmtx'M_j\pmtx + M)\right|\right)}\\
&\qquad \times \inf_{\theta\in\reals^d} e^{\theta'(p-\pvct) + \frac{1}{2}\left(V - 1_{j>0}\pmtx'M_j\pvct - \pmtx'\theta\right)'\left(P^H + 1_{j>0}\pmtx'M_j\pmtx + M\right)^{-1}\left(V - 1_{j>0}\pmtx'M_j\pvct - \pmtx'\theta\right)}
\end{split}
\end{equation}
Assume for now that $\pmtx$ is invertible. The optimizer is then
\begin{equation*}
\wh{\theta}_j = - 1_{j>0}M_j p + (\pmtx')^{-1}\left(V+(M+P^H)\pmtx^{-1}(\pvct-p)\right).
\end{equation*}
Market clearing necessitates
\begin{equation*}
\begin{split}
\Psi - \sum_{k=1}^J \alpha_k V_k &= \sum_{j=0}^J \alpha_j\wh{\theta}_j = \frac{1}{\ol{\gamma}}(\pmtx')^{-1}\left(V+(M+P^H)\pmtx^{-1}\pvct\right)\\
&\qquad - \frac{1}{\ol{\gamma}}(\pmtx')^{-1}\left(M+P^H + \pmtx'\left(\ol{\gamma}\sum_{k=1}^J \alpha_j M_j\right)\pmtx\right)\pmtx^{-1}p,
\end{split}
\end{equation*}
which, using the $\ol{M}$ notation shows that
\begin{equation*}
p = \pmtx\left(P^H + M +\ol{M}\right)^{-1}\left(V+(M+P^H)\pmtx^{-1}\pvct + \ol{\gamma}\pmtx'\left(\sum_{k=1}^J \alpha_k V_k - \Psi\right)\right).
\end{equation*}
With the $\ol{V}$ notation this becomes
\begin{equation*}
p = \pvct + \pmtx\left(P^H + M +\ol{M}\right)^{-1}\left(V+\ol{V}\right).
\end{equation*}
Using that $V = -\ol{V} + (P^H+M+\ol{M})\pmtx^{-1}(p-\pvct)$ we obtain
\begin{equation*}
\wh{\theta}_j = -1_{j>0} M_j p + (\pmtx')^{-1}\left(-\ol{V} + \ol{M}\pmtx^{-1}(p-\pvct)\right) = -1_{j>0}M_j p + \ol{\gamma}\left(\Psi + \sum_{k=1}^J \alpha_k (M_kp-V_k)\right).
\end{equation*}
This implies
\begin{equation*}
\begin{split}
\gamma_j \wh{\pi}_j &= 1_{j>0} V_j + \wh{\theta}_j = 1_{j>0}(V_j - M_j p) + \ol{\gamma}\left(\Psi + \sum_{k=1}^J \alpha_k (M_k p - V_k)\right).
\end{split}
\end{equation*}
Lastly, we consider the certainty equivalent $\gamma_j \E^{j,-}$. Plugging the optimal $\wh{\theta}^j$ into \eqref{E:normal_pre_opt}, taking the negative logarithm and simplifying yields
\begin{equation*}
\begin{split}
\gamma_j \E^{j,-} &= \lambda_j + \frac{1}{2}\log\left(\left|1_d + \Sigma^H(1_{j>0}\pmtx'M_j\pmtx + M)\right|\right) + 1_{j>0}\left(\frac{1}{2}\pvct' M_j\pvct - p'V_j\right)\\
&\qquad  - \left(V-1_{j>0}\pmtx'M_j\pvct\right)'\pmtx^{-1}(p-\pvct) + \frac{1}{2}(p-\pvct)'(\pmtx')^{-1}\left(P^H + 1_{j>0} \pmtx'M_j\pmtx + M\right)\pmtx^{-1}(p-\pvct),\\
&= \lambda_j +\check{\lambda}_j + 1_{j>0}\left(\frac{1}{2}p'M_j p - p'V_j\right) - V'(P^H+M+\ol{M})^{-1}(V+\ol{V})\\
&\qquad + \frac{1}{2}(V+\ol{V})'(P^H + M + \ol{M})^{-1}(P^H+M)(P^H+M+\ol{M})^{-1}(V+\ol{V}),
\end{split}
\end{equation*}
where
\begin{equation}\label{E:sp_lambda_def}
\begin{split}
\lambda_j &\dfn \frac{1}{2}\log\left(\left|1_d + \Sigma^H(1_{j>0}\pmtx'M_j\pmtx + M)\right|\right).
\end{split}
\end{equation}
The result follows with
\begin{equation*}
\begin{split}
\E^{-}(v,\ol{v}) &= \frac{1}{2}(V+\ol{V})'(P^H + M + \ol{M})^{-1}(P^H+M)(P^H+M+\ol{M})^{-1}(V+\ol{V})\\
&\qquad  - V'(P^H+M+\ol{M})^{-1}(V+\ol{V}),
\end{split}
\end{equation*}
which is evidently linear quadratic in $V,\ol{V}$ with coefficients
\begin{equation}\label{E:sp_lq_coeff}
\begin{split}
\frac{1}{2}v'\mathbf{E} v&: \quad\mathbf{E} = -(P^H+M+\ol{M})^{-1}(P^H+M+2\ol{M})(P^H+M+\ol{M})^{-1},\\
\frac{1}{2}\ol{v}'\mathbf{E} \ol{v}&: \quad\mathbf{E} = (P^H+M+\ol{M})^{-1}(P^H+M)(P^H+M+\ol{M})^{-1},\\
v'\mathbf{E}\ol{v}&: \quad\mathbf{E} = -(P^H+M+\ol{M})^{-1}\ol{M}(P^H+M+\ol{M})^{-1}.
\end{split}
\end{equation}
Lastly, note that the equilibrium price $p$, strategies $\cbra{\wh{\pi}^j}$ and endowments $\cbra{\gamma_j \E^{j,-}}$ do not depend upon $\pmtx^{-1}$.  And, using the equilibrium price $p$ and writing $\pi^j = \wh{\pi}^j + \theta^j$ for $j=0,...,J$ one can show, without requiring $\pmtx$ to be invertible, that $\wh{\theta}^j = 0$. Therefore, the result follows without assuming $\pmtx$ is invertible.

\end{proof}


\section{Proofs from Section \ref{S:EqEst}}\label{AS:EqEst}

\begin{proof}[Proof of Lemma \ref{L:Lambda_ident}] We will prove \eqref{E:Lambda_ident} first.  The basic idea of the proof is that using \eqref{E:X_trans_density} and \eqref{E:H_def2}, for $m=1,...,n$ we may write $H_m = \mu(X_t,1-t) + \E_X + \E_m$ where $\E_X, \cbra{\E_m}$ are independent of each other and $\F^B_t$, and where $\E_X \sim N(0,\Sigma(1-t))$, $\E_m \sim N(0,E_m^{-1}), m=1,...,n$.  Then it is clear $\ol{H}_n$ is conditionally normal given $\F^B_t$, and the proof below uses test functions to identify the ratio of the pdfs at $t=t$ and $t=0$. To this end, take test functions $(\varphi_1,...,\varphi_n)$ and let  $A_t\in\F^B_t$.  Using \eqref{E:H_def}, \eqref{E:H_def2}, the independence lemma, the Markov property, and writing $p_C$ as the $N(0,C^{-1})$ pdf we find (throughout we omit the differentials to ease notation)
\begin{equation*}
\begin{split}
\expvs{1_{A_t}\prod_{m=1}^n \varphi_m(H_m) } &= \expvs{1_{A_t}\int\limits_{x,\cbra{y_m}}  p(t,X_t,T,x)\prod_{m=1}^n\varphi_m\left(x + y_m \right) p_{E_m}(y_m)}.
\end{split}
\end{equation*}
Setting $h_m  = x + y_m$, and using $p_{E_m}(h_m-x) = p_{E_m}(h_m)e^{-(1/2)x'E_mx + x'E_m h_m}$ allows us to re-write $\int_{x,\cbra{y_m}} \cdots$ as
\begin{equation*}
\begin{split}
&\int\limits_{x,\ol{h}_n}  \ p(t,X_t,T,x)\prod_{m=1}^n\varphi_m\left(h_m\right)p_{E_m}(h_m-x)\\
&\qquad = \int\limits_{\ol{h}_n}  \left(\prod_{m=1}^n p_{E_m}(h_m)\varphi_m(h_m)\right)\int\limits_{x}  p(t,X_t,T,x) e^{-\frac{1}{2} x'\left(\sum\limits_{m=1}^n E_m\right)x + x'\left(\sum\limits_{m=1}^n E_m h_m\right)},\\
&\qquad = \int\limits_{\ol{h}_n}  \left(\prod_{m=1}^n p_{E_m}(h_m)\varphi_m(h_m)\right)e^{\Lambda\left(t,X_t;M^n, V^n(\ol{h}_n)\right)}.
\end{split}
\end{equation*}
It therefore follows that
\begin{equation}\label{E:olH_n_ident}
\expvs{1_{A_t}\prod_{m=1}^n \varphi_m(H_m) } = \expvs{1_{A_t}\int\limits_{\ol{h}_n} \left(\prod_{m=1}^n\varphi_m(h_m)\right)\times \left(\prod_{m=1}^n p_{E_m}(h_m)\right)e^{\Lambda\left(t,X_t;M^n, V^n(\ol{h}_n)\right)}},
\end{equation}
and hence, as this works for all test functions, and because $\Lambda(t,X_t; M^n,V^n(\ol{h}_n)$ is $\F^B_t$ measurable,
\begin{equation*}
\condprobs{\ol{H}_n \in d\ol{h}_n}{\F^B_t} = \left(\prod_{m=1}^n p_{E_m}(h_m)\right)e^{\Lambda\left(t,X_t;M^n, V^n(\ol{h}_n)\right)} = \left(\prod_{m=1}^n p_{E_m}(h_m)\right)e^{\ell_n\left(t,X_t,\ol{h}_n\right)}.
\end{equation*}
\eqref{E:Lambda_ident} now follows by dividing the answer for $t=t$ by that for $t=0$.  We now establish \eqref{E:Lambda_ident_2}.  Here, we may write
\begin{equation*}
\begin{split}
&(G_k,\ol{H}_n) = \bigg(\mu(X_t,1-t) + \E_X+Y_k, \mu(X_t,1-t) + \E_X+\E_1,...,\mu(X_t,1-t) + \E_X+\E_{k-1},\\
&\quad  \mu(X_t,1-t) + \E_X + Y_k + \frac{1}{\alpha_k}C_k^{-1}Z_k,\mu(X_t,1-t) + \E_X+\E_{k+1},...,\mu(X_t,1-t) + \E_X + E_n\bigg),
\end{split}
\end{equation*}
where  $\E_X, \cbra{\E_m}_{m\neq k}$ are as above, and independent of $Y_k,Z_k$ which themselves are independent of $\F^B_t$.  This implies $(G_k,\ol{H}_n)$ is jointly normal given $\F^B_t$, which using well known properties of Gaussian random vectors, also implies $G_k$ is normally distributed given $\F^B_t, \sigma(\ol{H}_n)$.  The rest of the proof identifies this pdf using test functions.
\begin{equation*}
\begin{split}
&\expvs{1_{A_t} \varphi(G_k)\prod_{m=1}^n \varphi_m(H_m)}\\
&\ = \expvs{1_{A_t}\int\limits_{x,\cbra{y_m,z_m}} p(t,X_t,T,x)\varphi(x+y_k)\prod_{m=1}^n \varphi_m\left(x + y_m + z_m)\right)p_{C_m}(y_m) p_{D_m}(z_m)}.
\end{split}
\end{equation*}
Keeping $x,\cbra{y_m}$ fixed, and writing $h_m = x + y_m + z_m$, the integral may be re-written
\begin{equation*}
\begin{split}
&\int\limits_{x,\cbra{y_m},\ol{h}_n} p(t,X_t,T,x)\varphi(x+y_k)\prod_{m=1}^n \varphi_m\left(h_m\right)p_{C_m}(y_m) p_{D_m}\left(h_m - y_m - x\right).
\end{split}
\end{equation*}
For $m\neq k$, $\int_{y_m} p_{C_m}(y_m)p_{D_m}(h_m-y_m-x) = p_{E_m}(h_m-x)$. This leads to
\begin{equation*}
\begin{split}
&\int\limits_{x,y_k,\ol{h}_n} p(t,X_t,T,x)\varphi(x+y_k)\varphi_k(h_k)p_{C_k}(y_k)p_{D_k}\left(h_k-y_k-x\right)\prod_{m\neq k}^n \varphi_m(h_m) p_{E_m}\left(h_m - x\right).
\end{split}
\end{equation*}
Lastly, we set $g_k = x+y_k$ to obtain
\begin{equation*}
\begin{split}
&\expvs{1_{A_t}\varphi(G_k)\prod_{m=1}^n\varphi_m(H_m)}\\
&\ = \mathbb{E}\bigg[1_{A_t} \int\limits_{x,g_k,\ol{h}_n} p(t,X_t,T,x)\varphi(g_k)\varphi_k(h_k)p_{C_k}(g_k-x)p_{D_k}\left(h_k-g_k\right)\prod_{m\neq k}^n \varphi_m(h_m) p_{E_m}\left(h_m - x\right)\bigg],\\
&\ = \mathbb{E}\bigg[1_{A_t} \int\limits_{g_k,\ol{h}_n} \varphi(g_k)\varphi_k(h_k)p_{C_k}(g_k)p_{D_k}\left(h_k-g_k\right)\left(\prod_{m\neq k}^n \varphi_m(h_m) p_{E_m}\left(h_m\right)\right)\\
&\qquad\qquad \times e^{\Lambda\left(t,X_t; C_k -E_m + M^n, C_k g_k -E_k h_k + V^n(\ol{h}_n)\right)}\bigg].
\end{split}
\end{equation*}
Copying the arguments used to establish \eqref{E:olH_n_ident} gives
\begin{equation*}
\begin{split}
&\expvs{1_{A_t}\chi(t,X_t,\ol{H}_n)\prod_{m=1}^n\varphi_m(H_m)}\\
&\ = \expvs{1_{A_t}\int\limits_{\ol{h}_n} \left(\prod_{m=1}^n p_{E_m}(h_m)\varphi_m(h_m)\right)\chi(t,X_t,\ol{h}_n)e^{\Lambda\left(t,X_t;M^n, V^n(\ol{h}_n)\right)}}.
\end{split}
\end{equation*}
This implies
\begin{equation*}
    \condexpvs{\varphi(G_k)}{\F^B_t \vee \sigma(\ol{H}_n)} = \frac{\int\limits_{g_k} \varphi(g_k)p_{C_k}(g_k)p_{D_k}(h_k-g_k)e^{\Lambda\left(t,X_t; C_k -E_m + M^n, C_k g_k -E_k h_k + V^n(\ol{h}_n)\right)}}{p_{E_k}(h_k)e^{\Lambda\left(t,X_t;M^n, V^n(\ol{h}_n)\right)}}.
\end{equation*}
As this must hold for all test functions $\varphi$ we conclude
\begin{equation*}
\begin{split}
&\condprobs{G_k \in dg_k}{\F^B_t\vee\sigma(\ol{H}_n)}\\
&\ = \frac{p_{C_k}(g_k)p_{D_k}(h_k-g_k)e^{\Lambda\left(t,X_t; C_k -E_m + M^n, C_k g_k -E_k h_k + V^n(\ol{h}_n)\right)}}{p_{E_k}(h_k)e^{\Lambda\left(t,X_t;M^n, V^n(\ol{h}_n)\right)}}.
\end{split}
\end{equation*}
This establishes \eqref{E:Lambda_ident_2}.

\end{proof}


\section{Proofs from Section \ref{S:pce_construct}}\label{AS:last_period}

\begin{proof}[Proof of Proposition \ref{P:eq_last_period}]  The proof follows almost exactly  \cite[Propositions 6.6, 6.7]{detemple2022dynamic} and hence we just provide a sketch, focusing on agent $0$. Write $\wh{\pi}^{0,\ol{h}_N}$ as the optimal policy from Proposition \ref{P:abstract_cont_time}. As clearly $\wh{\pi}^{0,\ol{H}_N} \in \A^{N,0}_{t_N,1}$, it suffices to show it minimizes the conditional expectation in \eqref{E:vf_def_lp}.  To this end, let $\pi = \pi^{\ol{H}_N} \in\A^{N,0}_{t_N,1}$, and to stress the dependence, write $\We^{\pi}_{t_N,1} = \We^{\pi,\ol{H}_N}_{t_N,1}$. As $\pi^{\ol{H}_N}\in \A^{N,0}_{t_N,1}$ there is a set $\E^{\pi,\ol{H}_N}$ of full Lebesgue measure where $\pi^{\ol{h}_N} \in \A^{\ol{h}_N}_{t_N,1}$. Using \eqref{E:U_ident_1}, \eqref{E:intuit_meas} and \eqref{E:intuit_relation}, we  have almost surely
\begin{equation*}
    \begin{split}
        &\condexpvs{e^{-\gamma_0 \We^{\pi,\ol{H}_N}_{t_N,1}}}{\F^m_{t_N}} = 1_{\E^{\pi,\ol{H}_N}}\condexpvs{e^{-\gamma_0 \We^{\pi,\ol{H}_N}_{t_N,1}}}{\F^m_{t_N}},\\
        &\quad = 1_{\E^{\pi,\ol{H}_N}}e^{-\ell_N(t_N,X_{t_N},\ol{H}_N)}\left(\condexpvs{e^{-\gamma_0 \We^{\pi,\ol{h}_N}_{t_N,1}-\frac{1}{2}X_1'M^N X_1 + X_1'V^N(\ol{h}_N)}}{\F^{B}_{t_N}}\right)\bigg|_{\ol{h}_N=\ol{H}_N},\\
        &\quad \geq 1_{\E^{\pi,\ol{H}_N}}e^{-\ell_N(t_N,X_{t_N},\ol{H}_N)}\left(\condexpvs{e^{-\gamma_0 \We^{\wh{\pi},\ol{h}_N}_{t_N,1}-\frac{1}{2}X_1'M^N X_1 + X_1'V^N(\ol{h}_N)}}{\F^{B}_{t_N}}\right)\bigg|_{\ol{h}_N=\ol{H}_N},\\
        &\quad = 1_{\E^{\pi,\ol{H}_N}}\condexpvs{e^{-\gamma_0 \We^{\wh{\pi},\ol{H}_N}_{t_N,1}}}{\F^m_{t_N}} = \condexpvs{e^{-\gamma_0 \We^{\wh{\pi},\ol{H}_N}_{t_N,1}}}{\F^m_{t_N}}.
    \end{split}
\end{equation*}
The same argument (c.f. \cite[Proposition 6.6]{detemple2022dynamic}) yields the optimality of $\wh{\pi}^{k,\ol{H}_N,G_k}$ for agent $k=1,...,N$, where $\wh{\pi}^{k,\ol{h}_N,g_k}$ is from Proposition \ref{P:abstract_cont_time}.  As the clearing condition is satisfied on the signal realization level, it is satisfied on the signal level as well, establishing the existence of the PCE.  The that martingale measures take the form in  \eqref{E:U_mm}, \eqref{E:I_mm} follows immediately from \eqref{E:big_condexp_ident}. Lastly, that the markets are complete is a direct consequence of \eqref{E:last_period_price_function}, as therein $t\to P(1-t) + \ol{M}^N$ is almost surely of full rank for all $t\in [t_N,1]$.
\end{proof}


\begin{proof}[Proof of Lemma \ref{L:sp_ass_verify}]
We start with Assumption \ref{A:cont_time_verify}. For $n=2,...,N$ recall that $\wt{\E}^{0,-}_{t_n}$ is the effective endowment at time $t_n$ (immediately prior to the jump) implied by the gains from all future trading over $[t_n,1]$ for agent $0$. As agent $0$ is allowed not to trade, we know $\wt{\E}^{0,-}_{t_n} \geq 0$ and hence (see equation \ref{E:Ik_prejump_n_endow}) $\check{\E}^{0,-}_n(X_{t_n},\ol{H}_n) \geq -\ell_{n-1}(t_n,X_{t_n},\ol{H}_{n-1})$. Using \eqref{E:N_matrices_def}, \eqref{E:ell_def_0} and the explicit formula for $\Lambda$ in \eqref{E:Lambda_def} we know $-\ell_{n-1}(t_n,x,\ol{h}_{n-1})$ is linear-quadratic in $\mu(x,1-t_n)$ with quadratic matrix
\begin{equation*}
    \begin{split}
       &\frac{1}{2}P(1-t_n) - \frac{1}{2}P(1-t_n)\left(P(1-t_n)+M^{n-1}\right)^{-1}P(1-t_n)\\
       &\qquad = \frac{1}{2}P(1-t_n)\left(P(1-t_n)+M^{n-1}\right)^{-1}M^{n-1}P(1-t_n) \in \sdpos.
    \end{split}
\end{equation*}
Therefore, as $\mu(x,1-t_n)$ is linear in $x$, and as we have already shown $\check{\E}^{0,-}_n(x,\ol{h}_n)$ is linear quadratic, we deduce its quadratic matrix for $x$ satisfies Assumption \ref{A:cont_time_verify}. For the $k^{th}$ insider, $k\leq n-1$ the result follows because $M_k = C_k-E_k \in \sdpos$ and hence $\pmtx'(C_k-E_k)\pmtx$ is non-negative definite.

As for Assumption \ref{A:lq_sp}, again recall that for $n=2,...,N$, $\wt{\E}^{0}_{t_n}$ is the effective endowment at time $t_n$ (right after the jump) implied by the gains from all future trading over $[t_n,1]$ for agent $0$. As such, $\wt{\E}^{0,-}_{t_n} \geq 0$ and hence (see equation \ref{E:checkE_tn_def}) $\check{\E}^0_n(X_{t_n},\ol{H}_n) \geq -\ell_n(t_n,X_{t_n},\ol{H}_n)$. But, this implies for Lebesgue almost every $x,\ol{h}_n$ that, with $P_H = E_n$ in Assumption \ref{A:lq_sp},
\begin{equation*}
    \check{\E}^0_n(x,\ol{h}_n) + \frac{1}{2}h_n'E_n h_n \geq \frac{1}{2}h_n'E_n h_n - \ell_n(t_n,x,\ol{h}_n).
\end{equation*}
Using \eqref{E:N_matrices_def}, \eqref{E:ell_def_0} and the explicit formula for $\Lambda$ in \eqref{E:Lambda_def} the right side above is linear-quadratic in $h_n$ with quadratic matrix
\begin{equation*}
    \begin{split}
       &\frac{1}{2}E_n - \frac{1}{2}E_n\left(P(1-t_n) + M^n\right)^{-1}E_n\\
       &\qquad = \frac{1}{2}E_n\left(P(1-t_n) + M^n\right)^{-1}\left(P(1-t_n) + M^{n-1}\right) \in \sdpos.
    \end{split}
\end{equation*}
Therefore, as we have already shown $\check{\E}^0_n(x,\ol{h}_n)$ is linear quadratic, we deduce its quadratic matrix for $h_n$ satisfies Assumption \ref{A:lq_sp}. For the $k^{th}$ insider, $k\leq n-1$ the result follows because $M_k = C_k-E_k \in \sdpos$ and hence $\pmtx'(C_k-E_k)\pmtx$ is non-negative definite.

\end{proof}


\section{Proofs from Section \ref{S:CJ}}\label{AS:CJ}

\begin{proof}[Proof of Proposition \ref{P:agg_Sig_conv}]
Let us start with the convergence of $F^N_t$. To get an intuitive idea for the result, assume $t$ is such that $n^N_t\dfn 2^N\tau^{-1}(t) - 1$ is an integer for every $N$.  This implies
\begin{equation*}
F^N_t = \sum_{n=1}^{n^N_t} E^N_n = E^N_1 + \sum_{n=2}^{n^N_t} E^N_n .
\end{equation*}
Using the \eqref{E:H_def2} and Assumption \ref{A:converge} we have
\begin{equation*}
C^N_1 = p(0)C_I, \ \alpha^N_1 = \frac{\omega_1}{\gamma(0)},\ D^N_1 \approx \frac{2^N}{\dot{\tau}(0)}D_Z; \quad \Rightarrow\quad E^N_1 \approx p(0)C_I.
\end{equation*}
Next, using \eqref{E:ll_def}
\begin{equation}\label{E:temp_Eval}
\begin{split}
C^N_n &= p\left(\tau\left(\frac{n-1}{2^N}\right)\right)C_I,\  \alpha^N_n = 2^{-N}\frac{\omega\left(\tau\left(\frac{n-1}{2^N}\right)\right)}{\gamma\left(\tau\left(\frac{n-1}{2^N}\right)\right)},\ D^N_n \approx 2^N \frac{1}{\dot{\tau}\left(\frac{n-1}{2^N}\right)}D_Z;\\
&\quad \Rightarrow \quad E^N_n \approx 2^{-N} \frac{\lambda\left(\tau\left(\frac{n-1}{2^N}\right)\right)^2}{\dot{\tau}\left(\frac{n-1}{2^N}\right)}C_I D_Z C_I.
\end{split}
\end{equation}
This implies $\sum_{n=2}^{n^N_t} E^N_n \approx \left(\int_0^{\tau^{-1}(t)} \frac{\lambda\left(\tau\left(u\right)\right)^2}{\dot{\tau}\left(u\right)}du\right)C_I D_Z C_I$ which yields the result for $F^N_t \to F_t$. The proof below makes this precise for all times $t$ by showing the error terms vanish. Indeed, set $n^N_0 = 1$ and for $t\in (0,1)$ define $n^N_t$ to be such that $t \in (t^N_{n^N_t}, t^N_{n^N_t + 1}]$. For all $t\in [0,1]$,
\begin{equation*}
    \tau^{-1}(t)  - \frac{1}{2^N} \leq  \frac{n^N_t}{2^N} < \tau^{-1}(t),
\end{equation*}
so that $n^N_t \to \tau^{-1}(t)$ as $N\to\infty$ and $|n^N_t - \tau^{-1}(t)| \leq 2^{-N}$. Next, recall $r^N_n$ from \ref{E:dyadic}. Then
\begin{equation*}
\begin{split}
F^N_{t} &= \sum_{n=1}^{n^N_t} E^N_n = \sum_{n=1}^{n^N_t} C^N_n\left(1 + \frac{1}{(\alpha^N_n)^2}\left(C^N_n\right)^{-1}(D^N_n)^{-1}\right)^{-1}\\
&= p(0)C_I\left(1+\frac{\gamma(0)^2\tau(2^{-N})}{p(0)(\omega^N_1)^2}C_I^{-1}D_Z^{-1}\right)^{-1}\\
&\quad  + \sum_{n=2}^{n^N_t} p\left(\tau(r^N_{n})\right)C_I\left(1 + 2^{N}\frac{\gamma(\tau(r^N_{n}))^2 2^{N}(\tau(r^N_{n+1})-\tau(r^N_{n}))}{\omega(\tau(r^N_{n})^2p\left(\tau(r^N_{n})\right)} C_I^{-1}D_Z\right)^{-1},\\
&= p(0)C_I\left(1+\frac{\gamma(0)^2\tau(2^{-N})}{p(0)(\omega^N_1)^2}C_I^{-1}D_Z^{-1}\right)^{-1}\\
&\quad  + 2^{-N}\sum_{n=2}^{2^N \times \frac{n^N_t}{2^N}} \frac{\lambda\left(\tau(r^N_{n})\right)^2}{2^{N}(\tau(r^N_{n+1})-\tau(r^N_{n}))}C_I D_Z C_I\bigg(1 + \frac{2^{-N}\lambda(\tau(r^N_n))^2}{p(\tau(r^N_n))2^{N}(\tau(r^N_{n+1})-\tau(r^N_{n}))} D_Z C_I\bigg)^{-1}.
\end{split}
\end{equation*}
Assumption \ref{A:converge}  ensures
\begin{equation*}
    \lim_{N\to\infty} p(0)C_I\left(1+\frac{\gamma(0)^2\tau(2^{-N})}{p(0)(\omega^N_1)^2}C_I^{-1}D_Z^{-1}\right)^{-1} = p(0)C_I.
\end{equation*}
This gives the convergence of $F^N_0$. For $t \in (0,1)$, note that for $n\leq n^N_t$, $r^N_n \leq \tau^{-1}(t) + 2^{-N}$ is uniformly bounded below $1$ for $N$ large enough.  Therefore, Assumption \ref{A:converge} implies $\lambda(\tau(r^N_n))$ and $p(\tau(r^N_n))$ are uniformly bounded.  Additionally, by the mean value theorem $0 < 2^{N}(\tau(r^N_{n+1}) - \tau(r^N_n)) = \dot{\tau}(\xi^N_n)$ for $\xi^N_n \in (r^N_n, r^N_{n+1})$ and hence from Assumption \ref{A:converge} we see it is uniformly  above $0$.  Given this, we conclude by the standard approximation to the integral that
\begin{equation*}
    \begin{split}
    \lim_{N\to\infty} F^N_{t} &= p(0)C_I + \left(\int_0^{\tau^{-1}(t)} \frac{\lambda(\tau(u))^2}{\dot{\tau}(u)}du\right) C_I D_Z C_I.
    \end{split}
\end{equation*}
As for $J_{t}$, given the convergence of $F^N_{n^N_t}$ we need only analyze
\begin{equation*}
\sum_{n=1}^{n^N_t} E^N_n \left(Y^N_n + \frac{1}{\alpha^N_n}\left(C^N_n\right)^{-1}Z^N_n\right).
\end{equation*}
As for $F^N_t$ we start with a heuristic argument. $E^N_1 \approx p(0)C_I$, $Y^N_1 = p(0)^{-1/2}C_I^{-1/2} B^Y_1$ and
\begin{equation*}
\frac{1}{\alpha^N_1}(C_1^N)^{-1}Z^N_1 = \frac{\gamma(0)}{\omega_1 p(0)} C_I^{-1}D_Z^{-1/2}\int_0^{2^{-N}} \sqrt{\dot{\tau}(u)}dB^Z_u \approx 0,
\end{equation*}
imply  $E^N_1(Y^N_1 + (1/\alpha^N_1)(C^N_1)^{-1}Z^N_1 \approx p(0)^{1/2}C_I^{1/2}B^Y_1$.  Next, for $n\geq 2$, using \eqref{E:ll_def}
\begin{equation*}
\begin{split}
Y^N_n &= p\left(\tau\left(\frac{n-1}{2^N}\right)\right)^{-1/2}C_I^{-1/2}(B^Y_n-B^Y_{n-1}),\\
\frac{1}{\alpha^N_1}(C_1^N)^{-1}Z^N_1 &= 2^N \frac{1}{\lambda\left(\tau\left(\frac{n-1}{2^N}\right)\right)} C_I^{-1}D_Z^{-1/2}\int_{\frac{n-1}{2^N}}^{\frac{n}{2^N}}\sqrt{\dot{\tau}(u)}dZ_u,
\end{split}
\end{equation*}
which, along with \eqref{E:temp_Eval} shows
\begin{equation*}
\sum_{n=2}^{n^N_t} E^N_n \left(Y^N_n + \frac{1}{\alpha^N_n}\left(C^N_n\right)^{-1}Z^N_n\right) \approx \int_0^{\tau^{-1}(t)}\frac{\lambda(\tau(u))}{\dot{\tau}(u)}C_I D_Z^{1/2}dZ^u.
\end{equation*}
We now make this argument precise. For $n=1$
\begin{equation*}
\begin{split}
&E^N_1 \left(Y^N_1 + \frac{1}{\alpha^N_1}\left(C^N_1\right)^{-1}Z^N_1\right) = p(0)C_I\left(1 + \frac{\gamma(0)^2\tau(2^{-N})}{p(0)(\omega^N_1)^2} C_I^{-1} D_{Z}^{-1}\right)^{-1}\\
&\qquad\qquad \times \left(p(0)^{-\frac{1}{2}}C_I^{-\frac{1}{2}}B^Y_1 + \frac{\gamma(0)}{p(0)\omega^N_1}C_I^{-1}D_{Z}^{-\frac{1}{2}}\int_0^{\tau(2^{-N})}\sqrt{\dot{\tau}(u)}dB^Z_u\right).
\end{split}
\end{equation*}
The stochastic integral on the right side above has mean $0$ and covariance matrix
\begin{equation*}
   \frac{\tau(2^{-N})}{(\omega_1^N)^2}\times C_I^{-1} D_Z^{-1} C_I^{-1},
\end{equation*}
so that under Assumption \ref{A:converge} we obtain the $L^2$ convergence
\begin{equation*}
\begin{split}
&\lim_{N\to \infty} E^N_1 \left(Y^N_1 + \frac{1}{\alpha^N_1}\left(C^N_1\right)^{-1}Z^N_1\right) = p(0)^{1/2} C_I^{1/2} B^Y_1.
\end{split}
\end{equation*}
This takes care of $J^N_0$. For $t\in (0,1)$ we also have the remaining sum
\begin{equation*}
\begin{split}
&\sum_{n=2}^{n^N_t} E^N_n \left(Y^N_n + \frac{1}{\alpha_{I^N_n}}\left(C^N_n\right)^{-1}Z^N_n\right)\\
&\qquad = \sum_{n=2}^{n^N_t} \frac{\lambda(\tau(r^N_{n}))^2}{2^N(\tau(r^N_{n+1})-\tau(r^N_n))} C_I D_Z C_I\left(1 + 2^{-N}\frac{\lambda(\tau(r^N_{n}))^2}{p(\tau(r^N_{n}))2^{-N}(\tau(r^N_{n+1})-\tau(r^N_n))} D_{Z} C_I\right)^{-1}\\
&\qquad \qquad \times \left(\underbrace{2^{-N}p(\tau(r^N_{n}))^{-\frac{1}{2}} C_I^{-\frac{1}{2}}\left(B^Y_{n+1}-B^{Y}_{n}\right)}_{1} + \underbrace{\frac{1}{\lambda(\tau(r^N_{n}))} C_I^{-1}D_Z^{-\frac{1}{2}}\int_{r^N_n}^{r^N_{n+1}}\sqrt{\dot{\tau}(u)}dB^Z_u}_{2}\right).
\end{split}
\end{equation*}
$t < 1$ implies that $\lambda(\tau(r^N_n)), p(\tau(r^N_n))$ are uniformly bounded, and that $2^N(\tau(r^N_{n+1}) - \tau(r^N_n))$ is uninformed bounded from above and away form $0$. Therefore,  the ``2'' term converges in $L^2$ to
\begin{equation*}
\int_0^{\tau^{-1}(t)} \frac{\lambda(\tau(u))}{\sqrt{\dot{\tau}(u)}}C_I D_Z^{\frac{1}{2}}dB^{Z}_u.
\end{equation*}
The ``1'' term has mean $0$ and covariance matrix
\begin{equation*}
\begin{split}
&2^{-2N}\sum_{n=2}^{n^N_t} \frac{\lambda(\tau(r^N_n))^4}{p(\tau(r^N_n))(2^N (\tau(r^N_{n+1}) - \tau(r^N_n)))^2}C_I D_Z C_I\left(1 + 2^{-N}\frac{\lambda(\tau(r^N_{n}))^2}{ p(\tau(r^N_n)) 2^{-N}(\tau(r^N_{n+1})-\tau(r^N_n))} D_{Z} C_I\right)^{-1}\\
&\quad \times C_I^{-1}\left(1 + 2^{-N}\frac{\lambda(\tau(r^N_{n}))^2}{ p(\tau(r^N_{n})) 2^{-N}(\tau(r^N_{n+1})-\tau(r^N_n))} D_{Z} C_I\right)^{-1}C_I D_Z C_I.
\end{split}
\end{equation*}
As everything is uniformly bounded, this goes to $0$ because
\begin{equation*}
    \begin{split}
        2^{-2N}\sum_{n=2}^{n^N_t} \frac{\lambda(\tau(r^N_n))^4}{p(\tau(r^N_n))(2^N (\tau(r^N_{n+1}) - \tau(r^N_n)))^2} \approx 2^{-N}\int_0^{\tau^{-1}(t)} \frac{\lambda(\tau(u))^4}{p(\tau(u))\dot{\tau}(u)}du,
    \end{split}
\end{equation*}
and $t< 1$. Putting everything together we obtain
\begin{equation*}
\lim_{N\to\infty} \sum_{n=1}^{n^N_t} E^N_n \left(Y^N_n + \frac{1}{\alpha^N_n}\left(C^N_n\right)^{-1}Z^N_n\right) = p(0)^{1/2} C_I^{1/2} B^Y_1 + \int_0^t \frac{\lambda(\tau(u))}{\sqrt{\dot{\tau}(u)}}C_I D_Z^{\frac{1}{2}}dB^{Z}_u,
\end{equation*}
giving the result.
\end{proof}


\begin{proof}[Proof of Lemma \ref{L:N_BM_lem}]  This result follows from \cite[Proposition 2.9]{MR3758346} but we offer a detailed proof using our notation. By construction of $\ell^N_n$ we know $t \to e^{\ell^N_n(t,X_t;\ol{h}^N_n)}$ is a $(\prob,\filt^B)$ Martingale, and a direct calculation using Ito's formula, \eqref{E:Lambda_def}, \eqref{E:theta_Nn_def} and the definition of $J^N_n$ shows
\begin{equation}\label{E:ZtprobjNn_def}
\frac{e^{\ell^N_n(T,X_T,\ol{h}^N_n)}}{e^{\ell^N_n(0,X_0,\ol{h}^N_n)}} = \E\left(\int_0^\cdot (\theta^{N,n}_u(j^N_n))'dB_u\right)_1,
\end{equation}
so that if we define $\tprob^{j^N_n}$ via the stochastic exponential above, then $B^{N,n}(j^N_n)_\cdot \dfn B_\cdot - \int_0^\cdot \theta^{N,n}_u(j^N_n)du$ is a $(\tprob^{j^N_n},\filt^B)$ Brownian motion. Next, from \cite[Theorem 3.2]{MR1775229} (with $Z^{F}\equiv 1$ therein), we know $B$ is a $\left(\prob^{\ol{H}^N_n},\filtg^{N,n}\right)$ Brownian motion, where using Lemma \ref{L:Lambda_ident} the martingale preserving measure $\prob^{\ol{H}^N_n}$  is defined through the Radon-Nikodym derivative
\begin{equation*}
\frac{d\prob^{\ol{H}^N_n}}{d\prob}\big|_{\F^{N,n}_T} = \frac{e^{\ell^N_n(0,X_0,\ol{H}^N_n)}}{e^{\ell^N_n(T,X_T,\ol{H}^N_n)}}.
\end{equation*}
This implies $B^{N,n} = B^{N,n}(J^N_n)$ is a continuous $\left(\prob^{\ol{H}^N_n},\filtg^{N,n}\right)$ semi-martingale, and hence by equivalence of the measures, $B^{N,n}$ is also a continuous $(\prob,\filtg^{N,n})$ semi-martingale. As $B^{N,n}$ has quadratic variation $t 1_d$ at all times $t$, it will be a $(\prob,\filtg^{N,n})$ Brownian motion provided it is a $(\prob,\filtg^{N,n})$ local martingale.  To this end, fix $0 \leq s < t\leq 1$, $A_s\in \F^B_s$ and a test function $\phi$.  As $J^N_n =  (F^N_n)^{-1}\sum_{m=1}^n E^N_m H^N_m$ is a function of $\ol{H}^N_n$, as an abuse of notation we will write $\theta^{N,n}(\ol{H}^N_n)$ and $\theta^{N,n}(J^N_n)$ interchangeably. We then have, for $s\leq t\leq 1$, Lemma \ref{L:Lambda_ident}
\begin{equation*}
\begin{split}
&\expvs{1_{A_s}\phi(\ol{H}^N_n)B^{N,n}_t} =\expvs{1_{A_s} \condexpvs{\phi(\ol{H}^N_n)\left(B_t - \int_0^t \theta^{N,n}_u(\ol{H}^N_n)du\right)}{\F^B_t}},\\
&\qquad = \int_{\ol{h}_n} \phi(\ol{h}^N_n)\left(\expvs{1_{A_s}\left(B_t - \int_0^t \theta^{N,n}_u(\ol{h}^N_n)du\right)\frac{e^{\ell^N_n(t,X_t,\ol{h}^N_n)}}{e^{\ell^N_n(0,X_0,\ol{h}^N_n)}}}\right)\prob\bra{\ol{H}_N \in d\ol{h}^N_n},\\
&\qquad = \int_{\ol{h}_n} \phi(\ol{h}^N_n)]\left(\expv{\tprob^{j^N_n}}{}{1_{A_s} B^{N,n}(j^N_n)_t}\right)\prob\bra{\ol{H}_N \in d\ol{h}^N_n},\\
&\qquad = \int_{\ol{h}_n} \phi(\ol{h}^N_n)]\left(\expv{\tprob^{j^N_n}}{}{1_{A_s} B^{N,n}(j^N_n)_s}\right)\prob\bra{\ol{H}_N \in d\ol{h}^N_n},
\end{split}
\end{equation*}
where the last equality used that $B^{N,n}(j^N_n)$ is a $(\tprob^{j^N_n},\filt^B)$ Brownian motion. Taking $t=s$ in the above chain of equalities gives the result.

\end{proof}


\begin{proof}[Proof of Lemma \ref{L:N_BM_lem2}]

We will prove the martingale property for $t^N_{n-1} < s < t^N_n < t < t^N_{n+1}$.  The other cases are similar. Here, we have
\begin{equation*}
\begin{split}
\condexpvs{B^N_t-B^N_s}{\G^N_s} &= \condexpvs{B_t - B_s - \int_s^{t^N_n} \theta^{N,n-1}(J^N_{n-1})_u du - \int_{t^N_n}^t \theta^{N,n}(J^N_n)_u du}{\G^N_s},\\
&=\condexpvs{\condexpvs{B_t - \int_{t^N_n}^t \theta^{N,n}(J^N_n)_u du}{\G^N_{t^N_t}} - B_s - \int_s^{t^N_n} \theta^{N,n-1}_u du}{\F^G_s},\\
&=\condexpvs{B_{t^N_n} - B_s - \int_s^{t^N_n} \theta^{N,n-1}(J^N_{n-1})_u du}{\G^N_s},\\
&= 0,
\end{split}
\end{equation*}
where we have used Lemma \ref{L:N_BM_lem} twice.
\end{proof}


\begin{proof}[Proof of Lemma \ref{L:theta_converge_p}]

Fix $t\in [0,1)$ and let $n^N_t$ be such that $t \in ( t^N_{n^N_t}, t^N_{n^N_t+1}]$. From \eqref{E:FNt_JNt_def}, \eqref{E:theta_Nn_def} and \eqref{E:theta_N_def} we see that
\begin{equation*}
    \theta^N_t = \sigma'e^{-(1-t)\kappa'}P(1-t)\left(P(1-t)+F^N_t\right)^{-1}F^N_t\left(J^N_t-\mu(X_t,1-t)\right).
\end{equation*}
The result follows immediately from Proposition \ref{P:agg_Sig_conv}.

\end{proof}


\begin{proof}[Proof of Proposition \ref{P:l2bdd}]
Throughout, $t<1$ is fixed. Next, from \eqref{E:theta_Nn_def} we obtain for $u \leq t$ that
\begin{equation*}
\left|\theta^{N,n}_u(J^N_n)\right|^2 = \left(J^N_n - \mu(X_u,1-u)\right)'M^{N,n}_u\left(J^N_n - \mu(X_u,1-u)\right),
\end{equation*}
where $M^{N,n}_u \in \sdpos$ is
\begin{equation*}
M^{N,n}_u = F^N_n(P(1-u)+F^N_n)^{-1}P(1-u)e^{-(1-u)\kappa}\Sigma e^{-(1-u)\kappa'}P(1-u)(P(1-u)+F^N_n)^{-1}F^N_n.
\end{equation*}
As $P(1-u)$ only blows up as $u\to 1$, it is clear for $u\leq t$ that there is a constant $C(t)$ such that $M^{N,n}_u \leq C(t) 1_d$ in the sense of positive definite matrices.  And hence (below, the constant $C(t)$ may change from line to line) we have
\begin{equation*}
\left|\theta^{N,n}_u(J^N_n)\right|^2 \leq C(t)\left( \left|\mu(X_u,1-u)\right|^2 + \left|J^N_n\right|^2\right).
\end{equation*}
Next, take a dyadic rational $r_0= 2^{-N_0}(n_0-1)$ so that $t<\tau(r_0)$. And, for $N>N_0$ let $n^N = 1 + 2^N r_0$ be such that $r_0 = 2^{-N}(n^N -1)$ and hence $\tau(r_0) = t^N_{n^N}$. Using \eqref{E:theta_N_def} we have
\begin{equation*}
    \begin{split}
        \expvs{\int_0^t \left|\theta^N_t\right|^2dt}  &\leq \expvs{\int_0^{t^N_{n^N}} \left|\theta^N_t\right|^2dt} =  \sum_{n=1}^{r_0 2^N} \int_{t^N_n}^{t^N_{n+1}} \expvs{\left|\theta^{N,n}_u(J^N_n)\right|^2} du.
    \end{split}
\end{equation*}
Using \eqref{E:X_trans_density} we know $\mu(X_u,1-u)  = \mu(X_0,1) + \E_X$ where $\E_X \sim N(0, e^{-(1-u)\kappa}\Sigma(u) e^{-(1-u)\kappa'})$. This means
\begin{equation*}
\expvs{\left|\mu(X_u,1-y)\right|^2} = \mu(X_0,1)'\mu(X_0,1) + \mathrm{Tr}\left(e^{-(1-u)\kappa}\Sigma(u) e^{-(1-u)\kappa'}\right),
\end{equation*}
and hence, because $J^N_n \sim N(0,F^N_n)$ we have
\begin{equation*}
\expvs{\left|\theta^{N,n}_u(J^N_n)\right|^2} \leq C(t)\left( 1 + \mathrm{Tr}(F^N_n)\right),
\end{equation*}
so that
\begin{equation*}
    \begin{split}
        \expvs{\int_0^t \left|\theta^N_t\right|^2dt}  &\leq C(t)\left(\tau(r_0) +  \sum_{n=1}^{r_0 2^N} \left(t^N_{n+1}-t^N_n\right)\mathrm{Tr}(F^N_n)\right).
    \end{split}
\end{equation*}
The result essentially follows using the arguments in the proof of Proposition \ref{P:agg_Sig_conv} which showed (see \eqref{E:temp_Eval})
\begin{equation*}
F^N_n = E^N_1 + \sum_{m=2}^n E^N_m  \approx p(0)C_1 + 2^{-N} \frac{\lambda\left(\tau\left(\frac{m-1}{2^N}\right)\right)^2}{\dot{\tau}\left(\tau\left(\frac{m-1}{2^N}\right)\right)} C_I D_Z C_I \leq C(t),
\end{equation*}
because for $m\leq n$ and $n\leq r_0 2^N$ we know $(m-1)/2^N \leq r_0$ which depends upon $t$ but is strictly bounded below $1$. Therefore,
\begin{equation*}
    \begin{split}
        \expvs{\int_0^t \left|\theta^N_t\right|^2dt}  &\leq C(t)\tau(r_0),
    \end{split}
\end{equation*}
giving the result. To make this fully precise, simply note that
\begin{equation*}
E^N_1 = \left(\frac{1}{p(0)}C_I^{-1} + \frac{\gamma(0)^2}{\omega_1^2 p(0)^2}\tau(2^{-N})C_I^{-1}D_Z^{-1}C_I^{-1}\right)^{-1} \leq p(0)C_I,
\end{equation*}
and
\begin{equation*}
\begin{split}
\sum_{m=2}^n E^N_m &= \sum_{m=2}^n \left( \frac{1}{p\left(\tau\left(\frac{m-1}{2^N}\right)\right)}C_I^{-1} + 2^{2N}\frac{\tau\left(\frac{m}{2^N}\right)-\tau\left(\frac{m-1}{2^N}\right)}{\lambda\left(\tau\left(\frac{m-1}{2^N}\right)\right)^2} C_I^{-1}D_Z^{-1}C_I^{-1}\right)^{-1},\\
&\leq \frac{1}{2^{2N}}\sum_{m=2}^n  \frac{\lambda\left(\tau\left(\frac{m-1}{2^N}\right)\right)^2}{\tau\left(\frac{m}{2^N}\right)-\tau\left(\frac{m-1}{2^N}\right)} C_I D_Z C_I
\end{split}
\end{equation*}
By the mean value theorem we have
\begin{equation*}
\tau\left(\frac{m}{2^N}\right)-\tau\left(\frac{m-1}{2^N}\right) = \frac{1}{2^N}\dot{\tau}(\zeta^N_n);\qquad \zeta^N_n \in \left[\frac{m-1}{2^N},\frac{m}{2^N}\right],
\end{equation*}
so that given our regularity assumptions and $m\leq n \leq 2^N r_0$ allow us to conclude
\begin{equation*}
\begin{split}
\sum_{m=2}^n E^N_m &\leq C(t)\frac{n}{2^N} C_I D_Z C_I \leq C(t) C_I D_Z C_I.
\end{split}
\end{equation*}
We can then use these bounds with $\mathrm{Tr}(A) \leq \mathrm{Tr}(B)$ if $B-A \in \sdpos$ to obtain the result.

\end{proof}


\begin{proof}[Proof of Proposition \ref{P:filt_converge}]

This essentially follows from \cite[Proposition 3.29]{neufcourtthesis}.  Indeed, by Lemma \ref{L:theta_converge_p} we know that $J^n_t\to J_t$ in probability for $t\in [0,1)$.  As $t\to J_t$ is continuous, we conclude that the finite dimensional distributions of $J^n$ converge to those for $J$ on $[0,1)$. \cite[Propositions 3.9 3.28, 3.29]{neufcourtthesis} then give the convergence for the natural filtrations of $(B,J^N)$ to that of $(B,J)$. Next,  \cite[Proposition 3.6]{neufcourtthesis} can easily be modified to allow for convergence in $L^p$.  Indeed, using the notation of \cite[Proposition 3.6]{neufcourtthesis}, if $Y^n\subseteq Z^n$ and $Y$ are sigma-algebras such that $Y^n\stackrel{L^p}{\to} Y$ and $Y^n\subseteq Z^n$, then for $B\in Y$ we have for $1<p\leq 2$ that
\begin{equation*}
    \begin{split}
    \expvs{\left| \condexpvs{1_B}{Z^n}  - 1_B\right|^p} &\leq \expvs{\left| \condexpvs{1_B}{Z^n}  - 1_B\right|^2}^{p/2},\\
    &\leq \expvs{\left| \condexpvs{1_B}{Y^n}  - 1_B\right|^2}^{p/2},\\
    &\leq 2^{\frac{(2-p)p}{2}} \expvs{\left| \condexpvs{1_B}{Y^n}  - 1_B\right|^p}^{p/2}.
    \end{split}
\end{equation*}
Similarly, if $p>2$ then
\begin{equation*}
    \begin{split}
    \expvs{\left| \condexpvs{1_B}{Z^n}  - 1_B\right|^p} &\leq 2^{p-2}\expvs{\left| \condexpvs{1_B}{Z^n}  - 1_B\right|^2},\\
    &\leq 2^{p-2}\expvs{\left| \condexpvs{1_B}{Y^n}  - 1_B\right|^2},\\
    &\leq 2^{p-2}\expvs{\left| \condexpvs{1_B}{Y^n}  - 1_B\right|^p}^{p/2}.
    \end{split}
\end{equation*}
Therefore, \cite[Proposition 3.6]{neufcourtthesis} allows us to say that $\filtg^{N}$ converges to the natural filtration of $(B,J)$ in $L^p$ and hence to $\filtg_0$ as it the natural filtration of $(B,J)$ augmented by $\prob$ null sets.
\end{proof}


\begin{proof}[Proof of Theorem \ref{T:limit_bm}]

Fix $t\in (0,1)$. Lemma \ref{L:theta_converge_p} that shows that for $u\in[0,t]$, $\theta_u^N \to \theta_u$ in probability, and hence $\theta^N\to\theta$ in $\prob\times\Leb_{[0,t]}$ probability. Additionally, Proposition \ref{P:l2bdd} shows $\cbra{\theta^N}$ is bounded in $L^2(P\times \Leb_{[0,t]})$. Therefore, Fatou's lemma implies $\theta \in L^2(P\times \Leb_{[0,t]})$ and by uniform integrability we know $\theta^{N}\to \theta$ in $L^p(P\times \Leb_{[0,t]})$ for any $1\leq p < 2$. Now, fix $s < u\leq t$ and let $A_s \in \G_{0,s}$. Then
\begin{equation*}
    \begin{split}
        &\expvs{1_{A_s}\left(B_u - B_s - \int_s^u \theta_v dv\right)}  = \expvs{\left(1_{A_s} - \condexpvs{1_{A_s}}{\G^N_s}\right)\left(B_u - B_s - \int_s^u \theta_v dv\right)}\\
        &\qquad + \expvs{\condexpvs{1_{A_s}}{\G^N_s}\int_s^u \left(\theta^N_v - \theta_v\right) dv} + \expvs{\condexpvs{1_{A_s}}{\G^N_s}\left(B_u - B_s -\int_s^u \theta^N_v dv\right)}.
    \end{split}
\end{equation*}
The third term on the right is $0$ as can be seen by conditioning on $\G^N_s$ and using Lemma \ref{L:N_BM_lem2}. The second term goes to $0$ because  $\theta^{N}\to \theta$ in $L^1(P\times \Leb_{[0,t]})$. The first term also goes to $0$ because for $1 < q < 2$ and $p>1$ such that $1/p + 1/q = 1$ we have
\begin{equation*}
    \begin{split}
        &\expvs{\left|1_{A_s} - \condexpvs{1_{A_s}}{\G^N_s}\right|\left|B_u - B_s - \int_s^u \theta_v dv\right|}\\
        &\qquad \leq \expvs{\left|1_{A_s} - \condexpvs{1_{A_s}}{\G^N_s}\right|^p}^{1/p}  \expvs{\left|B_u - B_s - \int_s^u \theta_v dv\right|^{q}}^{1/q}.
    \end{split}
\end{equation*}
The first term goes to $0$ as $\G^N \to \G$ in $L^p$. The second term is finite as $q<2$.  The $(\prob,\filtg_0)$ martingale property for $B_\cdot - \int_0^\cdot \theta_u du$ readily follows.  Lastly, we extend to $\filtg$, using \cite[VI.Theorem 2]{MR2273672}.
\end{proof}


\begin{proof}[Proof of Proposition \ref{P:teq1}]
By changing variables we have
\begin{equation*}
    \int_0^1 \mathrm{Tr}\left(\left(\Sigma(1-t) + F_t^{-1}\right)^{-1}\right)dt = \int_0^1 \mathrm{Tr}\left(\left(\Sigma(1-t) + F_t^{-1}\right)^{-1}\right)\dot{\tau}(t)dt.
\end{equation*}
Under Assumption \ref{A:converge} this integral is always finite over $[0,1-\eps]$ for any small $\eps>0$. Thus, we only focus on the integral over $[1-\eps,1]$. Next, recall $F_{\tau(t)}$ from Proposition \ref{P:agg_Sig_conv}
\begin{equation*}
\begin{split}
F_{\tau(t)} =  p(0)C_I + \left(\int_0^{t}\frac{\lambda(\tau(u))^2}{\dot{\tau}(u)}du\right)C_I D_{Z} C_I \rdfn A + g(t)B.
\end{split}
\end{equation*}
With this notation,
\begin{equation*}
    \begin{split}
    &\int_{1-\eps}^1 \mathrm{Tr}\left(\left((F_{\tau(t)})^{-1} + \Sigma(1-\tau(t))\right)^{-1}\right)\dot{\tau}(t)dt\\
    &\qquad = \int_{1-\eps}^1 \frac{g(t)\dot{\tau}(t)}{1+(1-\tau(t))g(t)}\times \mathrm{Tr}\bigg( \bigg(\frac{1+\Sigma(1-\tau(t))A}{1+(1-\tau(t))g(t)} + \frac{(1-\tau(t))g(t)B}{1+(1-\tau(t))g(t)}\\
    &\qquad\qquad \times\frac{1}{1-\tau(t)}\Sigma(1-\tau(t))\bigg)^{-1}\times \left(\frac{1}{g(t)}A + B\right)\bigg)dt
    \end{split}
\end{equation*}
It is straightforward to show that for $t\in [1-\eps,1]$
\begin{equation*}
    \mathrm{Tr}\bigg( \bigg(\frac{1+\Sigma(1-\tau(t))A}{1+(1-\tau(t))g(t)} + \frac{(1-\tau(t))g(t)B}{1+(1-\tau(t))g(t)}\frac{1}{1-\tau(t)}\Sigma(1-\tau(t))\bigg)^{-1}\times \left(\frac{1}{g(t)}A + B\right)\bigg)
\end{equation*}
is uniformly bounded above $0$ and from above. Thus,
\begin{equation*}
    \int_{1-\eps}^1 \mathrm{Tr}\left(\left((F_{\tau(t)})^{-1} + \Sigma(1-\tau(t))\right)^{-1}\right)\dot{\tau}(t)dt <\infty \quad \Longleftrightarrow \quad \int_{1-\eps}^1 \frac{g(t)\dot{\tau}(t)}{1+(1-\tau(t))g(t)}dt < \infty.
\end{equation*}
This proves the first statement. As for $\wt{B}$, assume $Q<\infty$ so that $\wt{B}_1$ is well defined and adapted to $\filtg_0,\filtg$. Next let $s < t$ and $A_s$ be a set in either $(\G_0)_s)$ or $\G_s$. Then, for any $s < u <1$, using the martingale property (c.f. Theorem \ref{T:limit_bm}) over $[0,u]$ we deduce
\begin{equation*}
    \begin{split}
        \expvs{1_{A_s}(B_1-B_s - \int_s^1 \theta_v dv)} = \expvs{1_{A_s}\left(B_1-B_u - \int_u^1 \theta_v dv\right)} \leq  C\expvs{(B_1-B_u)^2 + \int_u^1 |\theta_v|^2 dv}
    \end{split}
\end{equation*}
The result follows taking $u\uparrow 1$.
\end{proof}

\bibliographystyle{siam}

\bibliography{master}


\end{document}